\documentclass[a4paper,pra,amsmath,amssymb,floatfix,longbibliography,aps,twocolumn,superscriptaddress,accepted=2026-02-20]{quantumarticle}
\pdfoutput=1
\usepackage{amsfonts,color,physics,enumerate}
\usepackage{dsfont} 
\usepackage{latexsym} 
\usepackage{mathtools}
\usepackage{tabularx}

\usepackage[utf8]{inputenc}
\usepackage[english]{babel}
\usepackage[T1]{fontenc}

\usepackage[numbers,sort&compress]{natbib}

\usepackage{amsmath,amsfonts,amssymb,amsthm,graphics,graphicx,epsfig,bbm}
\usepackage{graphicx}
\usepackage{lipsum}
\usepackage{subfigure}
\usepackage{amsmath}
\usepackage{float} 
\usepackage{epsfig}
\usepackage{dcolumn}
\usepackage{bm}
\usepackage{color}
\usepackage{enumerate}
\usepackage{epstopdf}
\usepackage{amssymb}
\usepackage{amstext}
\usepackage{latexsym}
\usepackage{hyperref}
\usepackage{psfrag}
\usepackage{soul,xcolor}
\usepackage[normalem]{ulem}
\usepackage{physics}
\usepackage{footnote}
\usepackage{multirow}
\usepackage{appendix}
\usepackage{mathtools}
\usepackage{xspace}

\begin{document}
\newtheorem{Proposition}{Proposition}[section]

\title{Impact of Clifford operations on non-stabilizing power and quantum chaos}

\author{Naga Dileep Varikuti}
\email{dileep.varikuti@unitn.it}

\author{Soumik Bandyopadhyay}
\email{sbandyop@ictp.it}
\thanks{Present address: The Abdus Salam International Center for Theoretical Physics, Strada Costiera 11, 34151 Trieste, Italy}

\author{Philipp Hauke}
\email{philipp.hauke@unitn.it}

\affiliation{Pitaevskii BEC Center, CNR-INO and Dipartimento di Fisica, Universit\`a di Trento, Via Sommarive 14, Trento, I-38123, Italy}	
\affiliation{INFN-TIFPA, Trento Institute for Fundamental Physics and Applications, Via Sommarive 14, Trento, I-38123, Italy}

\newtheorem*{theorem*}{Theorem}
\setstcolor{red}

\begin{abstract}

Non-stabilizerness, alongside entanglement, is a crucial ingredient for fault-tolerant quantum computation and achieving a genuine quantum advantage. Despite recent progress, a complete understanding of the generation and thermalization of non-stabilizerness in circuits that mix Clifford and non-Clifford operations remains elusive. While Clifford operations do not generate non-stabilizerness, their interplay with non-Clifford gates can strongly impact the overall non-stabilizing dynamics of generic quantum circuits. In this work, we establish a direct relationship between the final non-stabilizing power and the individual powers of the non-Clifford gates, in circuits where these gates are interspersed with random Clifford operations. By leveraging this result, we unveil the thermalization of non-stabilizing power to its Haar-averaged value in generic circuits. As a precursor, we analyze two-qubit gates and illustrate this thermalization in analytically tractable systems. Extending this, we explore the operator-space non-stabilizing power and demonstrate its behavior in physical models. Finally, we examine the role of non-stabilizing power in the emergence of quantum chaos in brick-wall quantum circuits. 
Our work elucidates how non-stabilizing dynamics evolve and thermalize in quantum circuits and thus contributes to a better understanding of quantum computational resources and of their role in quantum chaos.

\end{abstract}

\maketitle

\newtheorem{theorem}{Theorem}[section]
\newtheorem{corollary}{Corollary}[theorem]
\newtheorem{lemma}[theorem]{Lemma}
\def\endproof{\hfill$\blacksquare$}

\section{Introduction}

Preparing resourceful states and operators with non-classical correlations is essential for fault-tolerant quantum computation and achieving a true quantum advantage over classical methods \cite{aspuru2005simulated, RevModPhys.81.865, datta2008quantum, campbell2017roads, preskill2018quantum, RevModPhys.91.025001}. 
Entanglement has been a primary quantum resource with applications ranging from quantum metrology \cite{giovannetti2004quantum, giovannetti2006quantum, RevModPhys.89.035002, hauke2016measuring, RevModPhys.90.035005}, quantum teleportation \cite{bennett1993teleporting}, and quantum error correction \cite{shor1995scheme, RevModPhys.87.307, roffe2019quantum} to quantum optimization algorithms \cite{hauke2015probing, PhysRevA.111.022434, PhysRevA.109.012413, chen2022much}. While an important asset in various quantum protocols, entanglement alone does not fully determine the computational power of quantum systems. In fact, stabilizer states generated by Clifford circuits---even if highly entangled---remain classically simulable \cite{gottesman1998heisenberg, aaronson2004improved}. Thus, the presence of non-stabilizerness, also called "magic," along with entanglement, is essential for universal fault-tolerant quantum computation \cite{bravyi2005universal, bravyi2012magic}. This insight motivated an avalanche of studies on the emergence of non-stabilizer dynamics in many-body quantum systems \cite{zhou2020single, liu2022many, niroula2024phase, catalano2024magic, bu2024complexity, campbell2011catalysis, turkeshi2023measuring, haug2023quantifying, PhysRevLett.131.180401, tarabunga2024nonstabilizerness, tarabunga2024magic, turkeshi2025magic, lopez2024exact, odavic2024stabilizer, hou2025stabilizer, andreadakis2025exact, dowling2025bridging, robin2025stabilizer, robin2024magic, tirrito2024anticoncentration, PRXQuantum.5.030332}. While non-stabilizerness is key to quantum advantage, a complete understanding of how it builds up and thermalizes \footnote[1]{As is common in this context \cite{jonnadula2020entanglement, d2014long}, we use the term \text{thermalization} to refer to the equilibration of non-stabilizing power to its Haar-averaged value.} in a quantum circuit 
remains outstanding.

In this work, we address this question by studying in detail the generation of non-stabilizerness in circuits that mix Clifford and non-Clifford operations. 
Our main theoretical framework is given by the non-stabilizer-generating power of the circuits under study. In general, the resource-generating powers of quantum evolutions provide a state-independent framework for understanding their ability to produce essential quantum resources. In this context, the entangling power of bipartite unitaries has been extensively studied, see, e.g., Refs.~\cite{zanardi2000entangling, nielsen2003quantum, jonnadula2017impact, jonnadula2020entanglement, lakshminarayan2001entangling, pal2018entangling, styliaris2021information, varikuti2022out} and references therein. The entangling power is defined as the average entanglement a unitary generates when acting on a typical product state \cite{zanardi2000entangling}. Similarly, the non-stabilizing power of a unitary, introduced in Ref.~\cite{leone2022stabilizer}, quantifies the average non-stabilizerness produced when acting on a typical stabilizer state. 
When a unitary has low entangling power, the evolution it generates is amenable to efficient classical simulation through the matrix product state framework \cite{verstraete2008matrix, schollwock2011density, perez2006matrix}. Likewise, a small non-stabilizing power of Clifford circuits interspersed with a few non-Clifford elements maintains the classical simulability \cite{PhysRevLett.116.250501, PRXQuantum.6.010337, liu2024classical}. 
Understanding the interplay between these two powers is key to identifying the limits of classical simulability \cite{gu2024magic, viscardi2025interplay, fux2024entanglement, frau2024nonstabilizerness}. Moreover, these two quantities are intimately related to out-of-time ordered correlators, a well-known diagnostic of quantum chaos \cite{zanardi2000entangling, varikuti2022out, shukla2022out, leone2022stabilizer}. Therefore, the mutual influence of entangling and non-stabilizing power is crucial for understanding quantum chaos and the build-up of complexity in many-body systems \cite{santra2025complexitytransitionschaoticquantum, chernyshev2025quantum, brokemeier2025quantum}.  

By definition, Clifford operations do not generate non-stabilizerness when applied to an arbitrary quantum evolution. 
This property motivates us to ask the following crucial question: \textit{How does a random Clifford operation contribute to generating non-stabilizerness when it is interspersed between two arbitrary non-Clifford unitaries?} We address this question with the help of rigorous analytical and numerical results. In particular, we show that the final non-stabilizing power displays an intimate---and functionally rather simple---connection with the individual powers of the non-Clifford unitaries. We then unveil the thermalization of non-stabilizing power to the Haar-averaged value under repeated insertions of random Clifford operations between arbitrary non-Clifford unitaries. Such interlacing of Clifford and non-Clifford gates naturally arises in settings like randomized benchmarking with random Clifford gates~\cite{magesan2012efficient}, highlighting the broader relevance of our findings to practical quantum protocols.
Further, we emphasize the interplay of the entangling and non-stabilizing powers in generating quantum chaos in quantum circuits. We do so by studying quantum chaos in brick wall quantum circuits where the two-qubit gates display smooth variation of entangling and non-stabilizing powers with interaction strength. Our findings highlight that quantum chaos emerges from the interplay of these quantities rather than from either quantity alone.

This work is structured as follows. In Sec.~\ref{sec-2}, we briefly review the definitions of entangling and non-stabilizing powers and summarize the main results of this work. We then detail the effect of random Clifford operations on the non-stabilizing power in Sec.~\ref{sec-impact}. Section~\ref{sec-impact-a} explores the thermalization of the non-stabilizing power and characterizes its associated fluctuations. In Sec.~\ref{sec-impact-b}, we introduce the operator-space non-stabilizing power and study its behavior with the help of chaotic and integrable Ising models. Then, in Sec.~\ref{sec-quant-chaos}, we construct minimally random brick-wall Floquet circuits and investigate the emergence of quantum chaos. Finally, we conclude this work in Sec.~\ref{sec-conclusion}. 
Technical details and supporting results are delegated to several appendices. 

\section{Background and Summary of Main Results}
\label{sec-2}
In this section, we briefly outline the main quantities of interest---the entangling power and the non-stabilizing power of quantum evolutions. While the former is defined for bipartite evolutions, the latter is defined for arbitrary quantum systems. To better guide the reader and to provide an overview of the contents of the article, we also summarize our main results.  

\subsection{Entangling power}
The entangling power of a bipartite unitary quantum evolution is defined as the average entanglement it generates when acting on typical product states \cite{zanardi2000entangling, nielsen2003quantum}. Consider a bipartite Hilbert space $\mathcal{H}_A \otimes \mathcal{H}_B$ with respective dimensions $d_A$ and $d_B$, and let $|\psi\rangle = |\phi_A\rangle \otimes |\phi_B\rangle$ be a product state in it. Then, for a bipartite unitary operator $U$ acting on $\mathcal{H}_A \otimes \mathcal{H}_B$, the entanglement generated in the state $U|\psi\rangle$ can be quantified by the purity of the reduced density matrix, $\rho_B = \text{Tr}_A\left(U|\psi\rangle\langle\psi|U^\dagger\right)$ \cite{zanardi2000entangling},
\begin{align}
\label{eq:entanglementpurity}
\mathcal{E}(U|\psi\rangle)=& 1-\text{Tr}_{B}(\rho^2_{B})\nonumber\\
=& 1-\text{Tr}\left[ U^{\otimes 2}\left( |\phi_A\phi_B\rangle\langle\phi_A\phi_B |  \right) U^{\dagger \otimes 2}S_{BB'} \right] \nonumber\\
=&2\text{Tr}\left[ U^{\otimes 2}\left(|\phi_A\phi_B\rangle\langle\phi_A\phi_B|\right)^{\otimes 2}U^{\dagger \otimes 2} P_{BB'} \right],\nonumber\\
\end{align}
with $P_{BB'}=(\mathbb{I}_{BB'}-S_{BB'})/2$, where $\mathbb{I}_{BB'}$ and $S_{BB'}$ are the identity and swap operators supported over the replicas of $\mathcal{H}_{B}$.
For brevity, we denote $\mathcal{H}_{A'}$ and $\mathcal{H}_{B'}$ to be the replica Hilbert spaces of $\mathcal{H}_{A}$ and $\mathcal{H}_{B}$, respectively. 

The entangling power of $U$ is defined as the Haar average over $|\phi_A\rangle$ and $|\phi_{B}\rangle$ of Eq.~\eqref{eq:entanglementpurity}, 
\begin{eqnarray}
e_{p}(U)&=&\overline{\mathcal{E}(U|\psi\rangle)}^{|\phi_A\rangle, |\phi_B\rangle}\nonumber\\
&=&2\text{Tr}\left( U^{\otimes 2}\bm{\Pi}_{AA'}\bm{\Pi}_{BB'}U^{\dagger \otimes 2}\Pi^{-}_{BB'} \right),
\end{eqnarray}
where $\Pi_{AA'}$ and $\Pi_{BB'}$ denote the second moments of the Haar random states in the Hilbert spaces $\mathcal{H}_{A}\otimes\mathcal{H}_{A'}$ and $\mathcal{H}_{B}\otimes\mathcal{H}_{B'}$, respectively and are given by \cite{renes2004symmetric}
\begin{equation}
\bm{\Pi}_{AA'}=\dfrac{\mathbb{I}_{AA'}+S_{AA'}}{d_{A}(d_{A}+1)}\;\text{ and }\; \bm{\Pi}_{BB'}=\dfrac{\mathbb{I}_{BB'}+S_{BB'}}{d_{B}(d_{B}+1)}.
\end{equation}
Interestingly, $e_p(U)$ is related to the operator entanglement $E(U)$ of $U$ through the relation \cite{zanardi2000entangling} 
\begin{eqnarray}  
e_p(U) = \frac{1}{E(S)} \left[ E(U) + E(US) - E(S) \right],  
\end{eqnarray}  
where $ e_p $ is normalized to remain within the range $[0,1]$ and $S$ denotes the SWAP operator over $\mathcal{H}_{A}\otimes\mathcal{H}_{B}$. Moreover, the entangling power shares an intimate connection with information scrambling through averaged out-of-time ordered correlators \cite{styliaris2021information, varikuti2022out, shukla2022out} and tripartite mutual information \cite{pawan}.

Note that $ e_p $ is invariant under local unitary transformations, meaning that  
$e_p((u_A \otimes u_B) U_{AB}) = e_p(U_{AB} (v_A \otimes v_B)) = e_p(U_{AB})
$. Another local unitary invariant, known as gate-typicality ($ g_t $), has been introduced as a complementary measure to $ e_p $ and is defined as \cite{jonnadula2017impact, jonnadula2020entanglement} 
\begin{equation}  
g_t(U) = \frac{1}{2E(S)} \left[ E(U) - E(US) + E(S) \right].  
\end{equation}
The gate-typicality contrasts gates with similar entangling power but different operator entanglement. As we shall see in the later sections, $g_t$ plays an equally important role as $e_p$ in the emergence of quantum chaos.

\subsection{Non-stabilizing power}
Similar to the entangling power, one can define the non-stabilizing power for arbitrary unitary operators. Let $\mathcal{G}_{N}$ denote the group of Pauli strings. Then, a state $|\psi\rangle$ on $N$ qubits is 
considered to be a stabilizer state if there exists a subgroup $\mathcal{S} \subset \mathcal{G}_N$ of size $|\mathcal{S}| = 2^N$, where every element $P \in \mathcal{S}$ satisfies $P|\psi\rangle = |\psi\rangle$, making $|\psi\rangle$ a simultaneous $+1$ eigenstate of all $P$ in $\mathcal{S}$ \cite{gottesman1998heisenberg}. The normalizers of the Pauli group constitute the Clifford group. The Clifford operations generate the stabilizer states when they act on standard computational basis states. The stabilizer dynamics are known to be efficiently classically simulable \cite{gottesman1998heisenberg, aaronson2004improved}. 
Given an arbitrary state $|\psi\rangle$, one can quantify its closeness to being a stabilizer state using the linear stabilizer entropy as follows \cite{leone2022stabilizer}: 
\begin{align}
\label{eq:nonstabpowerdef}
 \mathcal{M}(|\psi\rangle)=& 1-2^N\sum_{i=0}^{2^{2N}-1}\dfrac{1}{2^{2N}}\langle\psi |P_i|\psi\rangle^4\nonumber\\
 =&1-2^N\text{Tr}\left[ Q\left(|\psi\rangle\langle\psi |\right)^{\otimes 4} \right], 
\end{align}
where 
\begin{equation}
Q=\left( \dfrac{1}{2^{2N}}\sum_{i=0}^{2^{2N}-1}P^{\otimes 4}_i \right)    
\end{equation}
is a projector in $\mathcal{H}^{\otimes 4}$, i.e., $Q^2=Q$, and the set $\{P_i\}$ denote the set of all $N$-qubit Pauli strings. Non-stabilizerness has the following key properties \cite{leone2022stabilizer}: (i) $\mathcal{M}(|\psi\rangle)=0$ if and only if $|\psi\rangle$ is a stabilizer state, i.e., $|\psi\rangle$ is generated by the application of Clifford unitaries on $|0\rangle^{\otimes N}$, (ii) $\mathcal{M}(|\psi\rangle)$ is invariant under the action of arbitrary Clifford operations on $|\psi\rangle$, i.e., $\mathcal{M}(C|\psi\rangle)=\mathcal{M}(|\psi\rangle)$, and (iii) it is upper bounded by $\mathcal{M}(|\psi\rangle)\leq \log_{2}((d+1)/2)$.

Having defined the non-stabilizerness of quantum states, one can define the non-stabilizing power of a unitary. The average amount of non-stabilizerness that a unitary generates when it acts upon an arbitrary stabilizer state is \cite{leone2022stabilizer} 
\begin{eqnarray}\label{nonst2}
m_{p}(U)&=& \overline{\mathcal{M}(U|\psi\rangle)} \nonumber\\
&=&1-2^N\text{Tr}\left[ QU^{\otimes 4}\overline{\left(|\psi\rangle\langle\psi |\right)^{\otimes 4}}U^{\dagger \otimes 4} \right],
\end{eqnarray}
where the overline indicates the average over all the stabilizer states in the $N$-qubit Hilbert space. Using the invariance of $\mathcal{M}(|\psi\rangle)$, it is apparent from Eq.~(\ref{nonst2}) that the non-stabilizing power of a unitary remains invariant under the action of random Cliffords on it, i.e., $m_{p}(C_1UC_2)=m_{p}(U)$. Moreover, $m_{p}(U)=0$ if and only if $U$ is a Clifford unitary. Also, the average non-stabilizing power of Haar-random unitaries is 
\begin{equation}\label{non-stab-haar-exact}
\overline{m_p}=1-2^N\text{Tr}\left( Q \bm{\Pi}_{4} \right)=1-\dfrac{4}{2^N+3}\,, 
\end{equation} 
where $\bm{\Pi}_4$ is the fourth moment of the ensemble of stabilizer states in $\mathcal{H}^{2^N}$.   
Similar to the entangling power, the non-stabilizing power has also been shown to have intimate connections with information scrambling through the $8$-point OTOCs for a specific choice of the initial operators \cite{leone2022stabilizer}.
Apart from the stabilizer entropy, other measures of non-stabilizerness such as stabilizer fidelity, stabilizer extent \cite{bravyi2019simulation}, stabilizer rank \cite{bravyi2016trading, bravyi2019simulation, PhysRevLett.116.250501}, Wigner negativity, and mana \cite{veitch2014resource, pashayan2015estimating} have also been studied in recent years. In addition, the notion of non-stabilizerness has also been extended to the Heisenberg picture \cite{dowling2024magic}. We focus here on the stabilizer entropy, due to its information-theoretic nature and its known connections with quantum chaos diagnostics \cite{garcia2023resource, leone2022stabilizer, tirrito2024quantifying, leone2024stabilizer, oliviero2022measuring, tarabunga2023many}.

\subsection{Summary of Main Results}
This work primarily aims to understand the impact of random Clifford operators on the non-stabilizing power in generic quantum circuits. Our central result, detailed in Theorem~\ref{theo1} of Sec.~\ref{sec-impact}, can be stated as follows: When a random Clifford operator $C$ is sandwiched between two arbitrary non-Clifford operators $U$ and $V$, the final non-stabilizing power, on average, is 
\begin{equation}
\langle m_p(VCU)\rangle_{C}=m_p(U)+m_p(V)-\dfrac{m_{p}(U)m_p(V)}{\overline{m_p}}.  
\end{equation}
A complete derivation of the above equation is provided in Appendix~\ref{app-e}. Surprisingly, the averaging procedure decouples the non-stabilizing powers of $U$ and $V$, given by $m_p(U)$ and $m_p(V)$, respectively. This result can have far-reaching consequences, including enabling precise control over the generation of non-stabilizing power in experimental implementations. 

A key implication, as established in Corollary~\ref{corrollaryc} of Sec.~\ref{sec-impact-a}, is that the non-stabilizing power thermalizes exponentially with time to the Haar-averaged value $\overline{m_p}$ in circuits composed of interlaced Clifford and non-Clifford operations. While Clifford gates alone do not generate non-stabilizerness, their presence ensures thermalization of the non-stabilizing power, which may or may not occur in their absence. 
Consider a sequence of identical non-Clifford operations, denoted with $U$, interspersed with random Clifford gates. 
The non-stabilizing power of this sequence then decays to the Haar value, with a relaxation rate $\lambda$ that depends only on $m_p(U)$ and $\overline{m_p}$, according to the following simple form: 
\begin{eqnarray}
 \lambda=-\ln\left[ 1-\dfrac{m_p(U)}{\overline{m_p}} \right] .  
\end{eqnarray}
In addition, our result allows one to define the operator-space non-stabilizing power (OSNP) for an arbitrary quantum evolution---the average amount of non-stabilizing power a random Clifford unitary $ C $ acquires when it evolves in the Heisenberg picture under a quantum evolution $ U(t) $, given by $ \langle m_p(U^{\dagger}(t) C U(t)) \rangle_{C} $. We demonstrate this quantity for both integrable and chaotic Ising chains. The corresponding results are presented in Sec.~\ref{sec-impact-b}.

{Previous studies have suggested a 
non-zero amount of non-stabilizerness is necessary and sufficient for quantum chaos \cite{leone2021quantum}}. 
Our work further establishes that the collective behavior of non-stabilizing power, entangling power, and gate typicality of the gates governs the emergence of quantum chaos, as evidenced in Sec.~\ref{sec-quant-chaos}. 
To demonstrate this, we take minimally random brick-wall Floquet circuits having fixed interactions, where the randomness is incorporated through random single-qubit Clifford operations. As the interaction strength is varied, the circuit displays transitions from regular to chaotic regimes. As a precursor to our analyses, we also study the $m_p$ of two-qubit gates and show that it varies smoothly with the interaction parameters; see Appendices~\ref{app-a}–\ref{app-c} for details. The rest of this article is dedicated to deriving these results and discussing their implications.


\section{Impact of Interlaced Clifford Operations on Non-Stabilizing Power}
\label{sec-impact}

In this section, we address the central goal of our paper, which is to demonstrate how random Clifford operations affect the non-stabilizing power in generic quantum circuits. Given an arbitrary non-Clifford unitary $U$, we have $ m_{p}(U) = m_{p}(CU) $ for any Clifford operator $C$. This relation implies that if 
$ \mathcal{U} = CU $  acts on an initial stabilizer state $|\psi\rangle$, then $C$ does not contribute to the generation of the magic in it. However, repeated applications of $ \mathcal{U} $  may lead to different outcomes compared to those of $U$, that is,
\begin{equation}  
m_{p}(\mathcal{U}^2) = m_{p}(CUCU) = m_{p}(UCU) \neq m_{p}(U^2).
\end{equation}
The inequality generically holds { if neither $C$ nor $U$ is the identity operator}. It is a main aim of this work to provide a comprehensive understanding of the non-stabilizing power of a composite unitary obtained by sandwiching a Clifford unitary between two non-Clifford unitaries (see the setting in Fig.~\ref{fig:sch2}). 

\begin{figure}
    \centering
    \includegraphics[width=0.625\linewidth]{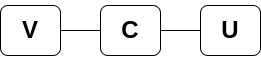}
    \caption{The setting considered in this work involves a random $ n $-qubit Clifford operation $ C $ sandwiched between two arbitrary non-Clifford unitaries, denoted as $ U $ and $ V $. We compute the average non-stabilizing power of this configuration, where the averaging is performed over the Clifford group supported over $ n $-qubits. }
    \label{fig:sch2}
\end{figure}

We first consider the case in which a random Clifford operation $C$ is sandwiched between two arbitrary non-Clifford unitaries, $U$ and $V$. The non-stabilizing power of the resulting composite unitary averaged over the Clifford group is then given by
\begin{align}\label{25}  
& \langle m_{p}(VCU) \rangle_{C} =1-2^N\int_{C} d\mu(C) \\  
& \hspace{2.em}\text{Tr}\left[ QV^{\otimes 4}C^{\otimes 4}U^{\otimes 4}\overline{\left(|\psi\rangle\langle\psi |\right)^{\otimes 4}} U^{\dagger \otimes 4}C^{\dagger \otimes 4}V^{\dagger \otimes 4} \right],  \nonumber 
\end{align}
where $\langle \cdots\rangle_{C}$ denotes the average over the Clifford group and $d\mu(C)$ is the invariant Haar measure associated with the Clifford group. Interestingly, $\langle m_p(VCU)\rangle_{C}$ is intimately connected to the non-stabilizing powers of $U$ and $V$, denoted by $m_p(U)$ and $m_p(V)$, respectively. This follows intuitively from Eq.~(\ref{25}), whose right-hand side remains invariant under the transformations $ U \to C_1 U C_2 $ and $ V \to C_3 V C_4 $ for arbitrary Clifford operations $ C_1, C_2, C_3$, and $ C_4 $. The non-stabilizing powers $ m_p(U) $ and $ m_p(V) $ 
naturally respect these transformations, suggesting a potential decoupling of $ \langle m_{p}(VCU) \rangle_{C} $ in terms of $ m_{p}(U) $ and $ m_{p}(V) $, i.e., $ \langle m_{p}(VCU) \rangle_{C} \sim f(m_{p}(U), m_{p}(V)) $. Indeed, we confirm this structure in the following theorem.

\begin{theorem}\label{theo1}
Let $U$ and $V$ be two arbitrary non-Clifford unitary operators supported over an $N$-qubit Hilbert space $\mathcal{H}=\mathbb{C}^{2^N}$ with non-stabilizing powers $m_p(U)$ and $m_p(V)$, respectively, and let $C$ be a Clifford operator sampled at random from the Clifford group according to its Haar measure. Then, the following relation holds:
\begin{equation}\label{mainres}
\langle m_{p}(VCU) \rangle_{C} = m_{p}(U)+m_{p}(V)-\dfrac{m_{p}(U)m_{p}(V)}{\overline{m_{p}}},  
\end{equation}
where $\overline{m_{p}}=\langle m_{p}(W)\rangle_{W}$ denotes the non-stabilizing power averaged over the Haar-random unitaries $W\in \mathcal{U}(2^N)$.  
\end{theorem}    

A detailed proof of this result is given in Appendix~\ref{app-e}. Although perhaps surprising, the simple relation for $\langle m_p(VCU) \rangle_C$ given in the above theorem arises as a direct consequence of the averaging performed over the Clifford group. To build intuition and validate the above theorem, it is useful to consider limiting cases. First, when $ U $ is the identity operator, we obtain $ \langle m_p(CV)\rangle_{C} = m_p(V) $. Conversely, if $ U $ (or $ V $) is chosen uniformly at random from the unitary group $ \mathcal{U}(d) $, the theorem yields $ \langle \langle m_p(VCU)\rangle_{C} \rangle_{U \in U(d)} = \overline{m_p} $, as expected.

Theorem~\ref{theo1} provides key insights into how random Clifford operations can enhance the non-stabilizing power. For instance, if one finds $U$ and $V$ such that $\langle m_{p}(UCV) \rangle_{C} > m_{p}(UV)$, it then follows that certain Clifford operations boost the non-stabilizing power beyond the serial application of $U$ and $V$. Moreover, if one considers $U$ and $V$ such that $m_p(U)=m_p(V)$, then Eq.~(\ref{mainres}) will become 
\begin{align}\label{sameU}
\langle m_{p}(VCU)\rangle_{C} 
=m_p(U)\left( 2-\dfrac{m_p(U)}{\overline{m_p}} \right).
\end{align}
If $m_p(U)<\overline{m_p}$, it follows that $\langle m_p(VCU)\rangle_{C} > m_p(U)$. In contrast, if $m_p(U)>\overline{m_p}$, we have $\langle m_p(VCU)\rangle_{C}<m_p(U)$. It may be surprising that, while in the former case the action of random Clifford operations enhances the final non-stabilizing power, in the latter case, the same operations diminish it. In addition, for $m_p(U)=\overline{m_p}$, we have $\langle m_p(VCU)\rangle_{C}=\overline{m_p}$. This indicates that $\overline{m_p}$ is the only non-trivial fixed point of Eq.~(\ref{sameU}). One can show that the same is true for Eq.~(\ref{mainres}). Consequently, {as we discuss in the next subsections}, for any $0<m_p(U)\neq \overline{m_p}$, one can repeatedly apply the non-Clifford unitaries followed by the random Clifford operations and reach $\overline{m_p}$.

{
While our main result in Eq.~(\ref{mainres}) characterizes the average effect of Clifford operators on the non-stabilizing power, it is also of interest to examine the fluctuations of $m_p(VCU)$ over random Clifford instances. This is important for assessing the typicality of the average behavior, and for determining whether Eq.~\eqref{mainres} is representative of individual Clifford realizations or dominated by rare instances. Here, we quantify the fluctuations through the variance of $m_p(VCU)$ over the entire Clifford group. We consider the simplest case where $U$ and $V$ are identical. In this setting, we find that 
\begin{align}\label{variance_main}
&\Delta^2 m_p(UCU)\nonumber\\ 
&\le \left[ 2m_p(U)-\dfrac{m^2_p(U)}{\overline{m_p}} \right] \left[ 1-2m_p(U)+\dfrac{m^2_p(U)}{\overline{m_p}} \right].
\end{align}
A detailed derivation of this inequality is given in Appendix~\ref{app-var}. The inequality becomes tight in the limit $m_p(U)=0$. Consequently, the fluctuations are minimal when $m_p(U)$ is small. Furthermore, when $m_p(U)=\overline{m_p}$, the inequality becomes $\Delta^2 m_p(UCU)\leq \overline{m_p}(1-\overline{m_p})$. Then, the right-hand side vanishes in the limit of large $N$ as $\overline{m_p}=1-4/(2^N+3)$ approaches $1$ with increasing $N$; see Eq.~(\ref{non-stab-haar-exact}). 
Taken together, these results indicate that the variance initially increases with $m_p(U)$, reaches a maximum at intermediate values, and decreases again for large $N$ as $m_p(U)$ approaches $\overline{m_p}$. We confirm this trend through numerical simulations presented later in this section. In the following, we discuss the thermalization of non-stabilizing power, along with other implications of Theorem~\ref{theo1}, including an analysis of the fluctuations.

}

\subsection{Evolution of Non-Stabilizing Power in Clifford-Interlaced Circuits}
\label{sec-impact-a}
The applicability of Eq.~(\ref{mainres}) extends beyond the single insertion of a random Clifford operation into quantum dynamics. In particular, Theorem \ref{theo1} enables an analytical derivation of the final non-stabilizing power when non-Clifford unitaries are repeatedly interspersed with independent random Clifford operations. For instance, let us take three arbitrary unitaries $ U_1, U_2, $ and $ U_3 $ interspersed with $C_1$ and $C_2$ drawn uniformly at random from the Clifford group. Upon averaging over both $C_1$ and $C_2$ independently, we get 
\begin{align} 
\langle m_{p}(&U_3C_2U_2C_1U_1) \rangle_{C_1, C_2}\nonumber\\
=\;&m_{p}(U_1)+m_{p}(U_2)+m_{p}(U_3)-\left[ m_{p}(U_1)m_{p}(U_2)\right.\nonumber\\
& \left.+m_{p}(U_1)m_{p}(U_3)+m_{p}(U_2)m_{p}(U_3) \right]/(\overline{m_{p}})\nonumber\\
&+ m_{p}(U_1)m_{p}(U_2)m_{p}(U_3)/\left(\overline{m_{p}}\right)^2. 
\end{align}
The right-hand side of the above equation can be rewritten in a simplified form as 
\begin{align}
\langle m_{p}(U_3C_2U_2&C_1U_1) \rangle_{C_1, C_2}=\overline{m_p}\left[ 1-\left( 1-\dfrac{m_p(U_1)}{\overline{m_p}} \right)\right. \nonumber\\
&\left.\left( 1-\dfrac{m_p(U_2)}{\overline{m_p}} \right)\left( 1-\dfrac{m_p(U_3)}{\overline{m_p}} \right) \right]. 
\end{align}
In the same way, for a generic quantum circuit consisting of a number $t$ of non-Clifford unitaries interlaced with independent random Clifford operations, the resulting non-stabilizing power can be determined from the following corollary of Theorem \ref{theo1}:

\begin{corollary}\label{corrollaryc}
Let $U^{(t)}=U_{t}C_{t-1}U_{t-1}\cdots C_1U_1$, where $\{C_{j}\}$ for all $1\leq j\leq t-1$ denote random Clifford operators drawn independently from the Clifford group according to its Haar measure, and $\{U_j\}$ for all $1\leq j\leq t$ are arbitrary non-Clifford unitaries, both supported over an $N$-qubit Hilbert space. Then, Theorem \ref{theo1} implies that 
\begin{eqnarray}\label{corrolary}
\left\langle m_p\left( U^{(t)} \right)\right\rangle_{\tilde{C}}=\overline{m_p}\left[ 1-\prod_{j=1}^{t}\left( 1-\dfrac{m_p(U_j)}{\overline{m_p}} \right) \right],     
\end{eqnarray}
where $\tilde{C}$ denotes the independent averaging over the Clifford group corresponding to $C_1, C_2, \cdots, C_{t-1}$.
\end{corollary} 

\begin{proof}
In order to prove the result in Eq.~(\ref{corrolary}), it is useful to consider the final unitary as $U^{(t)}=U_t C_{t-1}U^{(t-1)}$, where $U^{(t-1)}=U_{t-1}C_{t-2}U_{t-2}\cdots C_1U_1$. Then, Theorem ~\ref{theo1} implies the following:
\begin{align}\label{thermag}
\langle& m_{p}(U^{(t)})\rangle_{C_{t-1}} \nonumber\\
&= m_{p}(U_{t}) + m_{p}(U^{(t-1)}) - \dfrac{m_{p}(U_{t}) m_{p}(U^{(t-1)})}{\overline{m_{p}}}\nonumber\\
&=\overline{m_p}\left[ 1-\left( 1-\dfrac{m_p(U_t)}{\overline{m_p}} \right)\left( 1-\dfrac{m_p(U^{(t-1)})}{\overline{m_p}} \right) \right].\nonumber\\
\end{align}
We rewrite the above equation as 
\begin{align}
1-&\dfrac{ \langle m_p(U^{(t)})\rangle_{C_{t-1}} }{\overline{m_p}}\nonumber\\
&\hspace{2em}=\left( 1-\dfrac{m_p(U_{t})}{\overline{m_p}} \right)\left( 1-\dfrac{ m_p(U^{(t-1)})}{\overline{m_p}} \right).
\end{align}
Then, recursive applications of Theorem~\ref{theo1} at every time step lead to the following equation: 
\begin{eqnarray}\label{29}
 1-\dfrac{\langle m_p(U^{(t)})\rangle_{\tilde{C}}}{\overline{m_p}}  =  \prod_{j=1}^{t}\left( 1-\dfrac{m_p(U_{j})}{\overline{m_p}} \right). 
\end{eqnarray}
With a slight adjustment of terms in Eq.~(\ref{29}), we finally get
\begin{eqnarray*}
\left\langle m_p\left( U^{(t)} \right)\right\rangle_{\tilde{C}}=\overline{m_p}\left[ 1-\prod_{j=1}^{t}\left( 1-\dfrac{m_p(U_j)}{\overline{m_p}} \right) \right],    
\end{eqnarray*}
which gives the evolution of the non-stabilizing power under repeated applications of random Clifford and non-Clifford operators. 
\end{proof} 

Interestingly, the final non-stabilizing power depends only on the $m_p(U_j)$ rather than any other specific properties of the unitaries. It is essential that the interlaced Clifford elements at different time steps are independent, as correlations would prevent the decoupling of the non-stabilizing powers. Careful observation of Eq.~(\ref{corrolary}) reveals that in the limit of $t\rightarrow \infty$, the right-hand side converges to the Haar averaged value $\overline{m_p}$. Moreover, since $ |1 - m_p(U_j)/\overline{m_p}| \leq 1 $ for any $ U_j $, the term 
\begin{equation*}
\prod_{j=1}^{t} \left( 1 - \frac{m_p(U_j)}{\overline{m_p}} \right)
\end{equation*}
in Eq.~(\ref{corrolary}) is expected to decay exponentially with $t$. If there exists some $U_j$ such that $m_p(U_j)=\overline{m_p}$, then $\langle m_p(U^{(t)})\rangle_{\tilde{C}}=\overline{m_p}$ for any $t\geq j$. 
It is worth noting that doped Clifford circuits, where $ T$-gates are repeatedly interspersed with random Clifford operations, have been studied in previous works~\cite{leone2021quantum, leone2022stabilizer}. It has been shown that in such circuits, $m_p$ converges exponentially to $\overline{m_p}$. In contrast, our work explores a significantly more general setting that goes beyond $ T$-doped circuits. By analyzing arbitrary non-Clifford gates in conjunction with random Clifford layers, we uncover richer dynamics and deeper insights into the thermalization of non-stabilizing power in generic quantum circuits.

In the following, we study Eq.~(\ref{corrolary}) by considering two different scenarios involving interlacing Clifford and non-Clifford dynamics. First, we consider the case where all $ U_j $ are either identical or chosen alternatively from a fixed set of unitaries. In the second case, we select randomly generated $ U_j $ with varying non-stabilizing powers.

\subsubsection{\textbf{Case-1: All non-Clifford unitaries are the same}}
We first consider the case where identical non-Clifford operations ($U_j=U$, $\forall$ $ 0<j\leq t$) are interlaced with independent random Clifford unitaries. 
From Corollary~\ref{corrollaryc}, we obtain the closed-form expression 
\begin{align}\label{expo_magic}
\langle m_{p}(U^{(t)})\rangle_{\tilde{C}}=&\overline{m_{p}}\left[ 1-\left( 1-\dfrac{m_{p}(U)}{\overline{m_{p}}} \right)^{t} \right]\nonumber\\
=&\overline{m_p}\left[ 1-\exp\left\{ t\ln\left( 1-\dfrac{m_p(U)}{\overline{m_p}} \right) \right\} \right],
\end{align}
for $m_{p}(U)<\overline{m_{p}}$.
As it becomes evident from the above equation, $\langle m_{p}(U^{(t)})\rangle_{\tilde{C}}$ relaxes exponentially to the Haar-averaged value $\overline{m_p}$ with $t$. The corresponding relaxation rate is given by
\begin{equation}\label{eq:relax_rate}
\lambda=-\ln\left(1-\dfrac{m_{p}(U)}{\overline{m_{p}}}\right).   
\end{equation} 
On the contrary, when $m_p(U)>\overline{m_p}$, non-stabilizing power evolves as
\begin{align}\label{expo_magic_above}
\langle &m_{p}(U^{(t)})\rangle_{\tilde{C}}\nonumber\\
&=\overline{m_p}\left[ 1- (-1)^{t}\exp\left\{ t\ln\left(\dfrac{m_p(U)}{\overline{m_p}}-1 \right) \right\} \right],
\end{align}
indicating an exponentially damped oscillatory relaxation towards the Haar-averaged value. 
In all subsequent numerical analyses, we demonstrate our results by considering $m_p(U)<\overline{m_p}$.

\begin{figure}[h]
\includegraphics[scale=0.35]{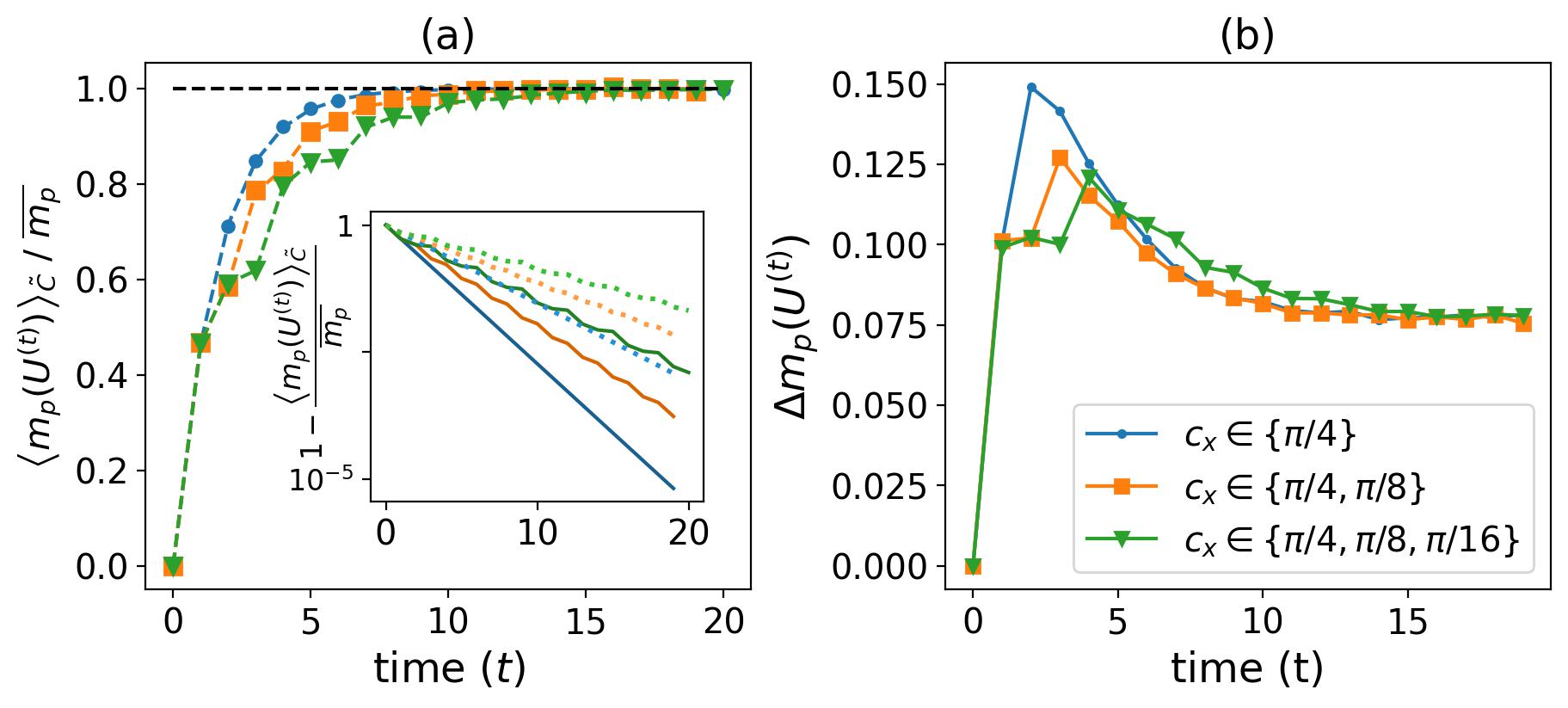}
\caption{\label{fig:nonstab1} Evolution of the non-stabilizing power $m_p$ in quantum circuits for $N=2$ qubits under three setups: (i) a fixed non-Clifford unitary interleaved with random Clifford gates (blue, circles), (ii) two distinct non-Clifford unitaries alternated with random Clifford gates (orange, squares), and (iii) a periodic sequence of three distinct non-Clifford unitaries interleaved with random Clifford gates (green, triangles). Each non-Clifford gate is of the form $U = \exp\{-i c_x \sigma_x \otimes \sigma_x /2\}$, with $c_x=\pi/4$ in (i), $c_x\in\{\pi/4,\pi/8\}$ in (ii), and $c_x\in\{\pi/4,\pi/8,\pi/16\}$ in (iii). Symbols denote numerical data averaged over $10^4$ circuit realizations, while dashed curves correspond to Eq.~(\ref{corrolary}). Inset: Semi-log plot showing exponential relaxation toward the Haar value for $N=2$ (solid lines) and for the same gates embedded in a four-qubit system as $\mathbb{I}_2 \otimes e^{-i c_x \sigma_x \otimes \sigma_x} \otimes \mathbb{I}_2$ (dotted lines). {(b) Fluctuations of $m_p(U^{(t)})$, quantified by the standard deviation $\Delta m_p(U^{(t)})$, for the three settings considered in (a). 
The fluctuations peak at intermediate times, consistent with the bound in Eq.~(\ref{variance_main}). 
We further observe that, at any fixed time, the fluctuation data coincide across all three cases whenever $\langle m_p(U^{(t)}) \rangle_{C}$ also coincide, indicating that the fluctuations depend mostly on the non-stabilizing power of the gates. For both plots, the numerical simulations are performed over $\sim 10^4$ circuit realizations. }}
\end{figure}

Now, we illustrate Eqs.~(\ref{expo_magic}) and (\ref{eq:relax_rate}) using an analytically tractable physical system. The simplest scenario is the two-qubit case; {see Appendix~\ref{app-a}}. Specifically, we take the unitary to be $U=\exp\{-ic_x\sigma_x\otimes\sigma_x/2\}$, whose non-stabilizing power, as given in Appendix~\ref{app-b}, is $m_p(U)=\sin^2(2c_x)/5$. Incorporating this value in Eq.~(\ref{expo_magic}), the dynamical generation of non-stabilizing power becomes
\begin{eqnarray}\label{ana}
\langle m_{p}(U^{(t)})\rangle_{\tilde{C}}=\overline{m_{p}}\left[ 1-\left( 1-\dfrac{\sin^2(2c_x)}{5\overline{m_{p}}} \right)^{t} \right], 
\end{eqnarray}
where $\overline{m_p}=1-4/7$ denotes the Haar-averaged value in the two-qubit Hilbert space. Choosing the parameter value that maximizes the single-unitary non-stabilizing power, $c_x=\pi/4$, the relaxation rate becomes
\begin{eqnarray}\label{ana2}
\lambda=-\ln\left(1-\dfrac{\sin^2(2c_x)}{5\overline{m_{p}}}\right)\approx 0.6286    \,.
\end{eqnarray}

To confirm Eqs.~(\ref{ana}) and (\ref{ana2}) numerically, we compute $\langle m_{p}(U^{(t)})\rangle_{\tilde{C}}$ by performing an average over $\sim 10^4$ independent realizations of the circuit. The corresponding results are shown in Fig.~\ref{fig:nonstab1}(a), indicated with blue color. The numerical simulations are performed for the first $20$ time steps. The first data point at $t=0$ corresponds to the identity operation. The numerical results in \ref{fig:nonstab1}(a) (blue dots) match the analytical predictions from Eq.~(\ref{ana}) (dashed curve) exactly up to negligible finite sampling fluctuations. Moreover, the numerical simulations yield $\lambda_{\text{num}}\approx 0.6281$, which matches with the analytically obtained exponent in Eq.~(\ref{ana2}). Further discussion on the evolution and equilibration of $\langle m_p(U^{(t)}) \rangle_{\tilde{C}}$ given in Eq.~(\ref{ana}), particularly in the regime where $m_p(U) \ll \overline{m_p}$, is provided in Appendix~\ref{app-F1} and Appendix~\ref{app-F2}.

We also consider another scenario where $ n $ distinct non-Clifford unitaries are applied in a cyclic sequence, interlaced with random Clifford operations. 
In this case, we still expect an exponential relaxation of $\langle m_p(U^{(t)})\rangle_{\tilde{C}}$ over the coarse-grained time scales. From Eq.~\eqref{corrolary}, one obtains the relaxation rate  
\begin{eqnarray}\label{rel2}
\lambda^{'}=-\dfrac{1}{n}\sum_{j=1}^{n}\ln\left( 1-\dfrac{m_p(U_j)}{\overline{m_p}} \right), 
\end{eqnarray}
where $m_p(U_j)<{\overline{m_p}}$ for all $j$.
In Fig.~\ref{fig:nonstab1}(a), we illustrate the behavior of such a cyclic sequence numerically for $n=2$ and $n=3$ (orange and green curves, respectively). For $n=2$, we fix the Euler parameters $c_x=\pi/4$ and $\pi/8$ alternatively. For $n=3$, the parameters are chosen to be $c_x=\pi/4$, $\pi/8$, and $\pi/16$. As in the previous case, the markers indicate the numerical results, while the dashed curves denote the analytically predicted behavior. The inset plot demonstrates the relaxation of $\langle m_p(U^{(t)})\rangle_{\tilde{C}}$ for all the cases considered so far. The results are shown with thick lines having the same color coding as above.
For all three cases, one can see a clear exponential relaxation, which for $n=2,3$ is only modified by a small periodic oscillation. 
For $n=2$, the relaxation rate is $\lambda'_{n=2}\approx 0.4414$, while for $n=3$, it is $\lambda'_{n=3}\approx 0.32$, in agreement with Eq.~(\ref{rel2}). 
These results suggest a rate that decreases with increasing period of application of the non-Clifford {gates}, given that $m_p(U_1)\geq m_p(U_2)\geq m_p(U_3)$. 

{
For the three setups considered in Fig.~\ref{fig:nonstab1}(a), we quantify the fluctuations around
$\langle m_p(U^{(t)}) \rangle$ using the standard deviation $\Delta m_p(U^{(t)})$. The corresponding numerical results are shown in Fig.~\ref{fig:nonstab1}(b). As anticipated from the inequality in Eq.~(\ref{variance_main}), the fluctuations are fully suppressed near $t=0$, where $m_p$ vanishes. As $m_p$ increases, the fluctuations grow, reach a maximum in the intermediate regime, and subsequently decrease. Moreover, at any fixed time, whenever $m_p$ coincides across all cases, the fluctuations appear to coincide. This indicates that the fluctuations are governed primarily by the non-stabilizing power $m_p$ of the unitaries, rather than by other gate-specific properties such as spectral statistics or operator entanglement. This is also evident from the analysis presented in the next subsection, see Fig.~\ref{fig:opspace}. 
 
}
 
For completeness, we also analyse the relaxation dynamics for a four-qubit case. We consider the non-Clifford unitaries of the form $U=\mathbb{I}_{2}\otimes \exp{-ic_x\sigma_x\otimes\sigma_x/2}\otimes\mathbb{I}_{2}$, again with the same set of parameters as before. Although $U$ acts trivially on two of the qubits, the Clifford operators can act on all four of them, thereby propagating the quantum information across the entire system. Through numerical extrapolation (see Appendix \ref{app-d}), we find that this unitary has the non-stabilizing power 
\begin{eqnarray}
m_p\left(\mathbb{I}_{2}\otimes \exp{-i\dfrac{c_x}{2}\sigma_x\otimes\sigma_x}\otimes\mathbb{I}_{2}\right)\approx\dfrac{\sin^2(2c_x)}{4.25}.
\end{eqnarray}
In this case, the Haar-averaged value is given by $\overline{m_p} = 1 - 4/(2^4 + 3)$. The corresponding relaxation dynamics are shown in the inset of Fig.~\ref{fig:nonstab1}(a), indicated by dotted curves with the same color coding as before. The relaxation is still exponential with the rates $\lambda'_{n=1}\approx 0.3539$, $\lambda'_{n=2}\approx 0.2576$ and $\lambda'_{n=3}\approx 0.1866$. These rates are comparatively slower than those observed in the two-qubit case. This is mainly due to the fact that $\sin^2(2c_x)/(5\overline{m}_{p, N=2})\geq \sin^2(2c_x)/(4.25\overline{m}_{p, N=4})$ for any fixed $c_x$, leading to a slower relaxation for $N=4$.

\begin{figure}[h]
\includegraphics[scale=0.35]{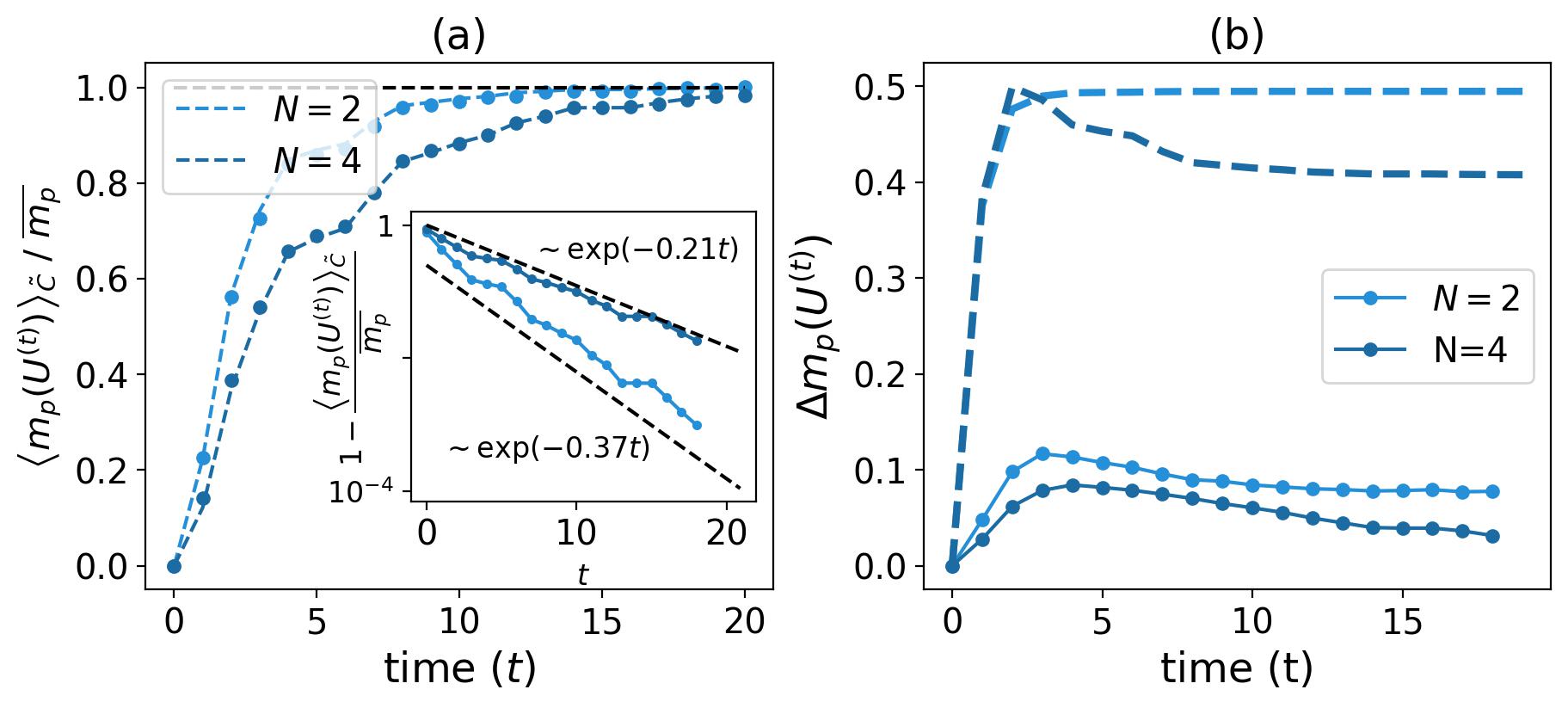}
\caption{\label{fig:nonstab2} Evolution and relaxation of the non-stabilizing power in circuits when non-Clifford unitaries with different values of $m_p$ are interspersed with random Clifford operations, for $N=2$ and $N=4$. For $N=2$, we use the same unitary form as before, i.e., $U=\exp\{-ic_x\sigma_x\otimes\sigma_x\}$, but now the $c_x$ are randomly drawn at every time step. For $N=4$, we use the same two-qubit gates embedded in a four-qubit Hilbert space. As shown in the inset, the overall relaxation remains, to a good approximation, exponential despite fluctuations around the mean values. For the same set of $c_x$ values, the case of $N=4$ shows a slower relaxation rate than the case of $N=2$. The numerical results are carried out for a single realization of the set of random $c_x$ values. {(b) Fluctuations of $m_p(U^{(t)})$ as quantified by the standard deviation $\Delta m_p(U^{(t)})$. Fluctuations become pronounced in the intermediate regime ($m_p\sim 0.7 \overline{m_p}$). In contrast, near $ m_p \sim \overline{m_p}$, the fluctuations decrease with system size, thereby validating the bound presented in the main text. The numerical simulations in both the panels are carried out over $\sim 10^4$ circuit realizations.}} 
\end{figure}

\subsubsection{\textbf{Case-2: Circuits with different non-Clifford operators}}
Here, we consider a more general case of Eq.~(\ref{corrolary}) by studying the interlacing of arbitrary unitaries having completely different non-stabilizing powers with random Clifford operations. 
{From the product structure in Eq.~(\ref{corrolary}), the exponential relaxation rate, over a number $t$ of time-steps follows} 
\begin{eqnarray}\label{avg_lambda}
\lambda_{\text{avg}}&=&-\dfrac{1}{t}\ln\left(1-\dfrac{\langle m_p(U^{(t)})\rangle_{\tilde{C}}}{\overline{m_p}}\right)\nonumber\\
&=&-\dfrac{1}{t}\sum_{j=1}^{t}\ln\left( 1-\dfrac{m_p(U_j)}{\overline{m_p}} \right).
\end{eqnarray}
This expression captures the cumulative effect of individual non-stabilizing powers on the relaxation rate of the final non-stabilizing power. 

For the setting considered here, we demonstrate the exponential relaxation of $\langle m_p(U^{(t)})\rangle_{\tilde{C}}$ using numerical simulations and compare with analytical predictions. The results are shown in Fig.~\ref{fig:nonstab2}(a) for two system sizes $N = 2$ and $4$, as in the previous case. For $N = 2$, we use the same form of the non-Clifford unitary as before, $U = \exp\{-i c_x \sigma_x \otimes \sigma_x /2\}$, with the parameter $c_x$ randomly drawn at each time step from the interval $[0, \pi/2]$. Then, Eq.~(\ref{avg_lambda}) predicts the average relaxation rate over the first $20$ time steps as $\approx 0.367$, which closely matches the numerically obtained value of $\approx 0.37$. 
For $N = 4$, the non-Clifford unitaries are, as above, of the form $U = \mathbb{I}_{2} \otimes \exp\{-i c_x \sigma_x \otimes \sigma_x /2\} \otimes \mathbb{I}_{2}$. Here, the analytical relaxation rate from Eq.~(\ref{avg_lambda}) is $\lambda \approx 0.2137$, which is in excellent agreement with the numerically fitted value of $\lambda\approx 0.21$. 

The corresponding numerical relaxation dynamics are displayed in the inset of Fig.~\ref{fig:nonstab2}(a). We notice that the rate is comparatively smaller than when $N=2$, which can be expected, since for the same gate parameters we found $m_p(U)/\overline{m_p}$ for $N=4$ (see Fig.~\ref{fig:N4magic}) to be smaller than that for $N=2$ (see Fig.~\ref{fig:nonst1}). Inserting these values into Eq.~\eqref{avg_lambda} implies the slower relaxation for $N=4$. {We note} that the numerical results in the current setting are carried out for a single realization of the set of $c_x$ values. 

{Figure~\ref{fig:nonstab2}(b) shows the fluctuations of $ m_p(U^{(t)}) $ for the two cases discussed above, using the same color coding. As anticipated, the fluctuations peak in the intermediate regime, while they vanish near $ t=0 $, corresponding to zero non-stabilizing power. At the opposite extreme, as $ m_p(U^{(t)}) $ approaches its Haar value, the fluctuations again decrease. We further observe that the fluctuations are suppressed with increasing system size. The figure also includes the corresponding bounds from Eq.~(\ref{variance_main}) for both cases, shown as dashed curves with the same color coding. Clearly, these bounds decrease with system size, and we expect that in the limit $ N \to \infty $, the bound vanishes in the vicinity of $m_p \sim \overline{m_p}$. }

\subsection{Operator-space non-stabilizing power}
\label{sec-impact-b}

The result in Theorem \ref{theo1} also provides a framework for studying the operator-space non-stabilizing power (OSNP) of quantum evolutions. The OSNP quantifies the average amount of non-stabilizing power that a unitary $U$ introduces into a random Clifford operator $C$ as it evolves under $U$ in the Heisenberg picture. 
That is, it characterizes how Clifford operators develop non-stabilizing properties over time. 
This can be formulated by replacing $V$ with $U^{\dagger}$ in Theorem \ref{theo1}. Noting that $m_{p}(U^{\dagger}) = m_{p}(U)$, we obtain for the OSNP of $U$ 
\begin{eqnarray}\label{opspacestab}
\text{OSNP}(U)&=&\langle m_{p}(U^{\dagger}CU)\rangle_{C} \nonumber\\
&=& m_{p}(U) \left( 2 - \dfrac{m_{p}(U)}{\overline{m_{p}}} \right).
\end{eqnarray}  
This equation captures how a given unitary transformation influences Clifford operators on average. 
Moreover, Eq.~(\ref{opspacestab}) quantifies the OSNP as a function of $m_{p}(U)$ itself. As discussed below Eq.~\eqref{sameU}, when $m_{p}(U)<\overline{m_{p}}$, we have $\text{OSNP}(U)>m_{p}(U)$. In contrast, if $m_{p}(U)>\overline{m_{p}}$, it follows that $\text{OSNP}(U)<m_{p}(U)$. 

We numerically demonstrate Eq.~(\ref{opspacestab}) for the Ising model with the following Hamiltonian:  
\begin{eqnarray}  
H = \sum_{j=0}^{N-1} \sigma^j_x \otimes \sigma^{j+1}_x + h_x \sum_{j=0}^{N-1} \sigma_x + h_y \sum_{j=0}^{N-1} \sigma_y,  
\end{eqnarray}  
under periodic boundary conditions. We consider two parameter sets, $(h_x, h_y) = (0.8090, 0.9045)$, representing chaotic dynamics~\cite{kim2014testing, varikuti2024unraveling}, and $(0,1)$, corresponding to integrable dynamics. {We keep the system size fixed at $N=8$}. Figure~\ref{fig:opspace}(a) shows the numerical results for $\langle m_{p}(U^{\dagger}(t)CU(t))\rangle_{C}$ as a function of time $t$, with $U(t)=\exp(-i H t/\hbar)$. 
The numerical simulations indicate that the OSNP rapidly approaches the Haar-averaged value $\overline{m_p}$ for both chaotic and integrable parameter regimes. Interestingly, the rate of growth in the integrable case deviates only slightly from that of the chaotic case. The insets show the relaxation dynamics of the OSNP on a semi-log scale by computing the quantity $1-\text{OSNP}(U)/\overline{m_p}$. {For the system size considered, the chaotic case relaxes with a rate of approximately $84$, compared to about $64$ in the integrable case, highlighting a quantitative separation between the two. Figure~\ref{fig:opspace}(b), on the other hand, shows the OSNP fluctuations over the Clifford group. The numerical results are consistent with those of the previous subsection, where a peak appears in the intermediate regime. Moreover, we observe no noticeable difference between the chaotic and integrable cases. This behavior can be understood from the earlier observation that this variance depends primarily on $m_p(U)$ rather than on other specific features of the dynamics. However, we note that the behavior of non-stabilizerness in ground states, as well as under quench protocols, may distinguish between integrable and chaotic Hamiltonian dynamics \cite{odavic2024stabilizer, viscardi2025interplay}. }

\begin{figure}
\includegraphics[scale=0.35]{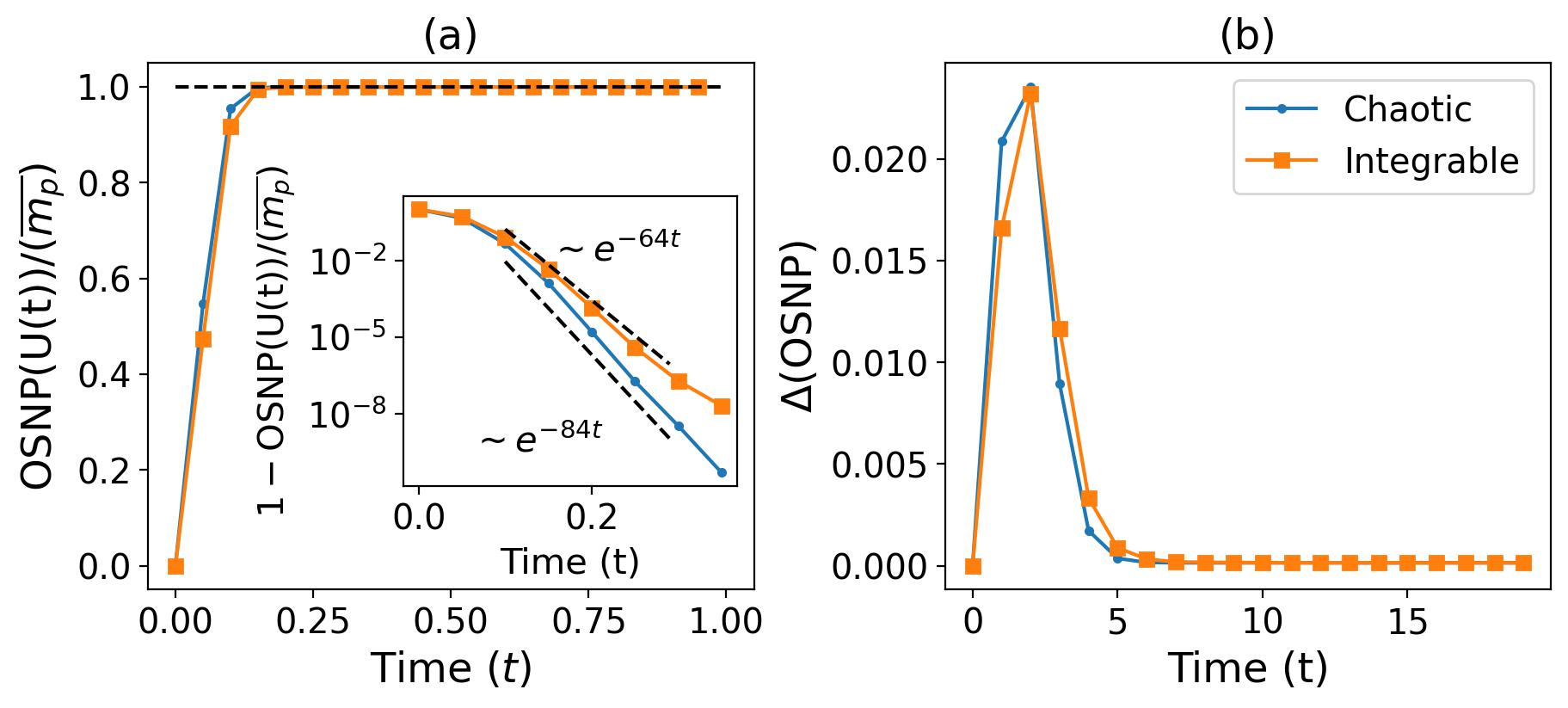}
\caption{\label{fig:opspace} {(a) The operator-space non-stabilizing power (OSNP) of $U(t)=\exp(-i H t/\hbar)$, where $H$ is the Hamiltonian of an Ising chain, for two different parameter regimes---chaotic (blue) and integrable (orange). The chain length considered is $N=8$. For the chaotic regime, the parameters are fixed at $h_x=0.8090$ and $h_y=0.9045$, and for the integrable case at $h_x=0$ and $h_y=1$. 
The data is normalized by the Haar averaged value $\overline{m_p}$. In both regimes, the OSNP rapidly approaches the Haar-averaged values, as indicated by the horizontal dashed line. Insets: the quantity $1-\text{OSNP}(U(t))/\overline{m_p}$ reveals an exponential relaxation for both parameter ranges. While the relaxation rate for the chaotic chain is found to be $\approx 84$, the corresponding rate in the integrable case is $\approx 64$, which is noticeably smaller. (b) Fluctuations in the OSNP, quantified by the standard deviation of the OSNP over the Clifford group. As predicted in the previous subsections, these fluctuations are significant in the intermediate regime for both cases. We observe no noticeable difference between the two cases. The numerical simulations are carried out for $\sim 10^3$ samples of the Clifford instances.
}}
\end{figure}

It may seem surprising that while the periodic boundary conditions imply translation symmetry, the OSNP still approaches its Haar-averaged value. This can be understood from the following reasoning. Suppose that as $ t \to \infty $, the non-stabilizing power of the time-evolution operator satisfies  $m_p(U(t)) = \overline{m_p} - \delta$
where $ 0 < \delta \ll \overline{m_p} $ accounts for symmetry constraints. Then, from Eq.~(\ref{opspacestab}), we obtain  
\begin{align}
\text{OSNP}(U) 
= \overline{m_p}\left[ 1- \left(\dfrac{\delta}{\overline{m_p}}\right)^2 \right]\,.
\end{align}
The above equation implies that if $ m_p(U(t)) $ itself deviates from $ \overline{m_p} $ by $ \delta $, the corresponding OSNP differs only by a significantly smaller second-order correction $\sim \delta^2 $. This ensures that the OSNP remains effectively indistinguishable from the Haar-averaged value despite the translation invariance of the system.  

Thus far, our focus has been on the impact of random Clifford operators on the non-stabilizing power and its thermalization toward the Haar-averaged value in generic quantum circuits. The next section turns to how this quantity, in conjunction with entangling power and gate typicality, governs the onset of quantum chaos. This interplay reveals that non-stabilizing power plays a critical role in the chaotic properties of quantum circuits.

\section{quantum chaos transitions in brick-wall Floquet circuits} 
\label{sec-quant-chaos}

In this section, we analyze the roles played by non-stabilizing power ($m_p$), entangling power ($e_p$), and gate-typicality ($g_t$) in the emergence of quantum chaos. To do so, we construct Floquet circuits using two-qubit gates as given in Eq.~\eqref{cartan} (see Appendix \ref{app-a}) up to random single-qubit Clifford operations, arranged in a brick-wall architecture. Floquet evolution has been widely used to study regular and chaotic properties of both single and many-body quantum systems \cite{bertini2018exact, kos2018many, bertini2019exact, chan2018spectral, lerose2021influence, dileep2024, varikuti2022out, srivastava2016universal, nag2014dynamical, dileep2024, PhysRevB.109.094207}. Recent studies suggest that random Clifford circuits over $ N $-qubits become quantum chaotic when interspersed with at least $ \Theta(N) $ non-Clifford resources such as $ T $-gates \cite{leone2021quantum}. However, understanding the role of the non-stabilizing power of individual gates in inducing quantum chaos remains an outstanding issue. Here, we demonstrate that quantum chaos in Floquet circuits arises from an intriguing interplay of $ e_p $, $ m_p $, and $g_t$ of the involved gates. Conversely, we demonstrate that quantum chaos can arise with minimal randomness, generated solely through single-qubit random Clifford operations. Our findings can have practical implications for quantum technologies that leverage or are affected by the presence of quantum chaos. These include NISQ-era quantum computations \cite{preskill2018quantum}, quantum simulations \cite{hauke2012can, PRXQuantum.1.020308, PhysRevResearch.3.033145, sahu2022quantum, heyl2019quantum, sieberer2019digital, chinni2022trotter,sauerwein2023engineering,uhrich2023cavity,baumgartner2024quantum, Kargi:21}, state tomography \cite{smith2013quantum, merkel2010random}, information recovery \cite{hayden2007black, yoshida2017efficient}, and the construction of unitary and state designs \cite{harrow2009random, brown2010convergence, cotler2023emergent, choi2023preparing, varikuti2024unraveling}
for randomized benchmarking \cite{benchmarking1, knill2008randomized, benchmarking2}, measurements \cite{vermersch2019probing, elben2023randomized}, and more.

{

To begin with, let us recall that any two-qubit unitary operator $U\in \text{SU}(4)$ can be written in the canonical form (usually referred to as the Cartan parametrization) using the Euler angles $\{c_x, c_y, c_z\}$ up to left and right multiplication of arbitrary single-qubit unitaries as \cite{khaneja2001time, kraus2001optimal, zhang2003geometric, bertini2019exact}
\begin{equation}\label{cartan-main}
U=\exp\left\{-i\sum_{j\in\{x, y, z\}}\dfrac{c_j}{2}\left(\sigma_{j}\otimes\sigma_{j}\right)\right\},
\end{equation}
where $c_j\in (0, \pi]$ and $\{\sigma_j\}_{j\in\{x, y, z\}}$ are the Pauli operators. Imposing local equivalence, i.e., requiring that two unitaries related by local transformations share the same set of Euler angles, restricts the ranges of the Euler angles to $ 0\leq c_z \leq c_y \leq c_x \leq \pi/2 $ \cite{zhang2003geometric}. A schematic illustrating the space of two-qubit unitaries embedded inside a tetrahedron is shown in Fig.~\ref{fig:schem}. The vertices of this geometry are locally equivalent to well-known two-qubit Clifford unitaries: Identity ($c_x,c_y,c_z=0, 0, 0$), CNOT ($\pi/2, 0, 0$), double-CNOT or \hbox{DCNOT} also known as iSWAP ($\pi/2$, $\pi/2$, 0), and SWAP ($\pi/2$, $\pi/2$, $\pi/2$) gates. When the single-qubit unitaries are Clifford gates, the tetrahedron’s vertices correspond to the two-qubit Clifford gates. An extensive treatment of $e_p$ and $g_t$ for these edges has been carried out in Ref.~\cite{jonnadula2020entanglement}. For the sake of completeness, we provide their explicit expression:
\begin{align}\label{epgt}
e_p(U)&= \dfrac{2}{3}\sum_{\mathrm{cyc}(x,y,z)} \sin^2 c_x \cos^2 c_y \,,\nonumber\\
g_t&=\sum_{i\in\{x, y, z\}} \sin^2(c_i). 
\end{align}
Appendix~\ref{app-b} presents the analysis of the corresponding $m_p$ along several edges.
}

Having outlined the two-qubit gate parametrization, we construct the Floquet circuits using gates selected from the edges of the tetrahedron shown in Fig.~\ref{fig:schem}. The Floquet operator for a single time step is pictorially given by
\begin{eqnarray}  
\mathcal{U} = \raisebox{-1.6cm}{\includegraphics[scale=0.15]{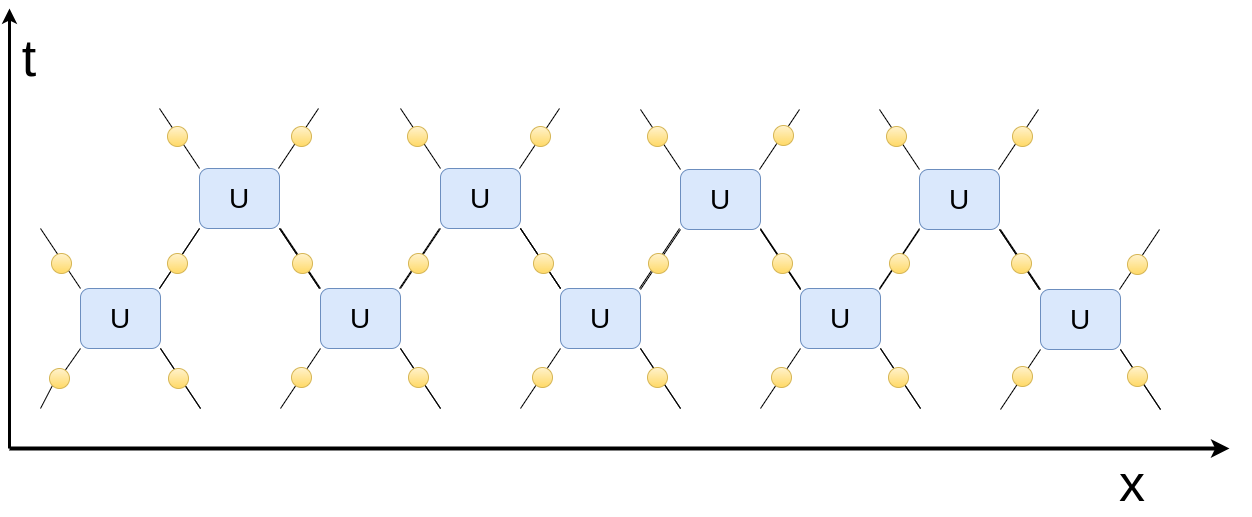}}  \label{eq:Floquet_unitary}
\end{eqnarray}  
where lines represent qubits (arranged for illustration along a spatial dimension $x$), $ U$ denotes a two-qubit gate with fixed Euler angles from Eq.~\eqref{cartan-main}, and yellow circles indicate random and independent single-qubit Clifford operations. To quantify quantum chaos in the above setting, we compute the average adjacent level-spacing ratio $\langle r \rangle$ of the Floquet circuit. This quantity is defined as \cite{oganesyan2007localization, atas2013distribution}
\begin{equation}
\langle r \rangle =\text{Avg}\left\{r_i\right\}_{i=1}^{2^N-2},\thickspace \textrm{where}\hspace{0.2cm} r_i=\dfrac{\min(d_i, d_{i+1})}{\max(d_i, d_{i+1})}, 
\end{equation}
and $d_i=\phi_{i+1}-\phi_{i}$ denotes the $i$-th eigenphase spacing, where $\phi_{i}$ are the eigenphases of the Floquet operator $\mathcal{U}$ depicted in Eq.~\eqref{eq:Floquet_unitary}. For sufficiently chaotic circuits, this quantity approaches the value of a completely Haar-random unitary, $\langle r\rangle_{\text{CUE}}=0.596543$ \cite{d2014long}. Here, CUE corresponds to circular unitary ensembles of complex random matrices \cite{dyson1962threefold}. In contrast, if the circuit is regular, $\langle r\rangle$ matches the value of the Poisson ensemble $\langle r\rangle_{\text{Poisson}}=0.386294$ \cite{d2014long}. As we encompass different edges of the tetrahedron sketched in Fig.~\ref{fig:schem}, the quantities $m_p$, $e_p$, and $g_t$ become the three control parameters against which the chaotic nature of the circuit can be examined. In Figs.~\ref{fig:magivsR}  and \ref{fig:magivsR1}, we present the trend of $\langle r\rangle$ along the edges, which we discuss in detail below. 
Among the considered edges, Id -- CNOT stands out, as the Floquet operators constructed from gates along this edge exhibit many-fold degeneracies, requiring a modified treatment through multiple brick-wall layers.

\begin{figure} 
\includegraphics[scale=0.35]{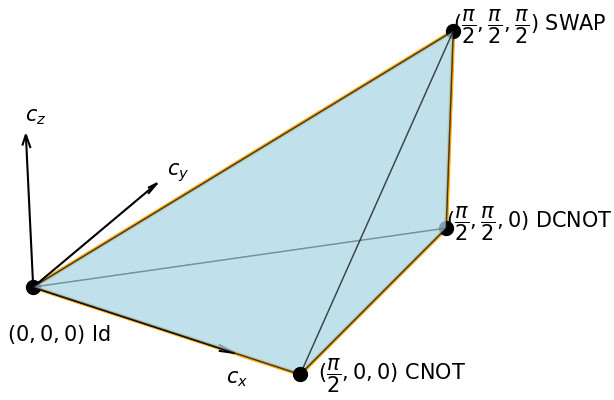}
\caption{\label{fig:schem} 
Space of two-qubit unitaries, illustrated as a tetrahedron parametrized by the Euler angles $c_x$, $c_y$, and $c_z$. 
Each point inside the tetrahedron uniquely determines the two-qubit unitaries up to single-qubit unitaries with $c_j\in[0, \pi/2)$ for all $j\in\{x, y, z\}$. In this work, we consider the unitaries that lie along the edges Id --- CNOT, CNOT --- DCNOT, SWAP --- DCNOT, and Id --- SWAP (or $S^{\alpha}$, the fractional powers of SWAP). These edges are marked with the orange color in the figure. } 
\end{figure}

\begin{figure*} 
\includegraphics[scale=0.365]{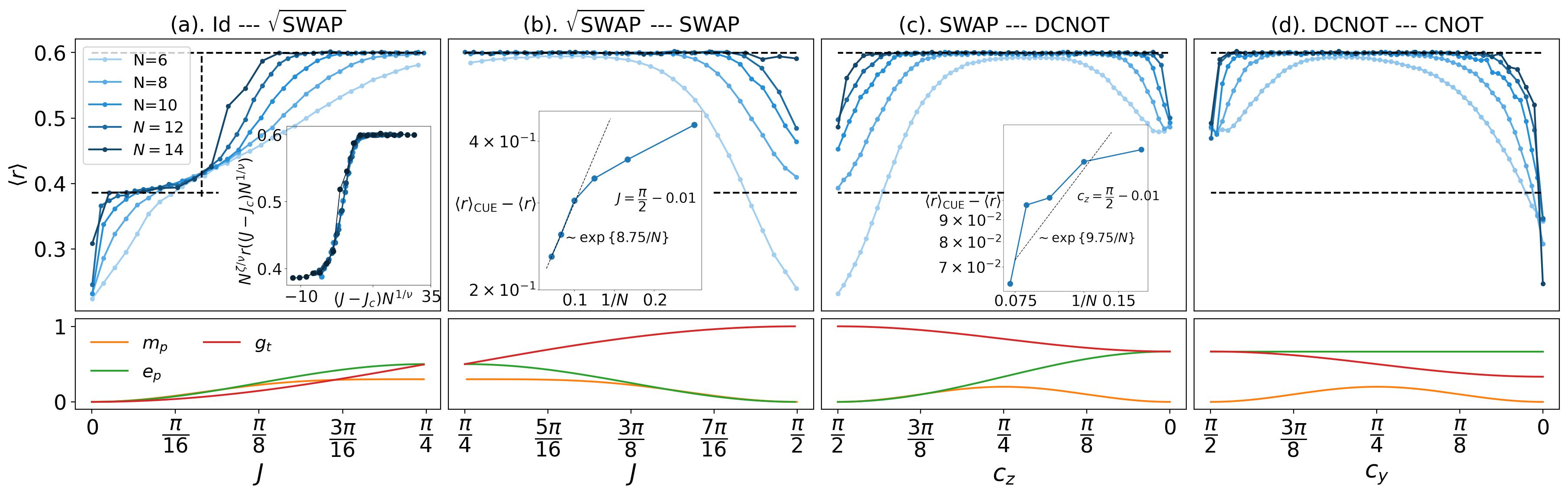}
\caption{\label{fig:magivsR} Emergence of quantum chaos in brick-wall Floquet unitaries constructed from two-qubit gates lying along the edges of the tetrahedron in Fig.~\ref{fig:schem}, analyzed through the average adjacent level spacing ratio, $ \langle r \rangle $. The lower panels denote the behaviour of non-stabilizing power $m_p$ (see also Fig.~\ref{fig:nonst1}), entangling power $e_p$, and gate typicality $g_t$ of the corresponding two-qubit unitaries considered in the above panels.   
(a) Id --- $\sqrt{\mathrm{SWAP}}$, 
(b) $\sqrt{\mathrm{SWAP}}$ --- SWAP, 
(c) SWAP --- DCNOT, 
(d), DCNOT --- CNOT. 
(a) Near $J\approx 0.26$, we notice a critical transition from Poisson to Wigner--Dyson statistics. The inset demonstrates the finite-size scaling analysis for the transition, yielding the critical point $J_c\approx 0.26$ and critical exponents $\nu \approx 0.65$ and $\zeta \approx 0.0014$. Here, $m_p$, $e_p$, and $g_t$ all grow monotonically from zero as shown in the bottom panel. 
(b) For small systems $\langle r\rangle$ is close to Poisson statistics as $J$ reaches $\pi/2$. The inset illustrates the behaviour of $\langle r\rangle$ near a point that is close to $J=\pi/2$, suggesting the system reaches Wigner-Dyson statistics in the thermodynamic limit as soon as $J$ deviates from exactly $\pi/2$. 
(c) The system experiences a quantum chaos transition as the Euler angle is perturbed away from the end points $c_z=\pi/2$ and $0$. In the vicinity of $c_z=\pi/2$ (SWAP), the behaviour of $\langle r\rangle$ is almost identical to that observed in panel (b) near $J=\pi/2$ (SWAP). The inset illustrates the behaviour of $\langle r\rangle$ near a point that is close to $c_z=0$. 
(d) Similarly, the system experiences a chaos transition at the endpoints $c_y=\pi/2$ and $0$. The behaviour of $\langle r\rangle$ near $c_y=\pi/2$ (DCNOT) closely mirrors that of panel (b) near $c_z=0$ (DCNOT). The data in the figure is shown for $N=6$, $8$, $10$, $12$, and $14$. Dashed lines represent $\langle r\rangle\approx 0.39$ and $\approx 0.60$, corresponding to regular and chaotic statistics, respectively. { For system sizes $N=6,8,$ and $10$, the numerical simulations were performed over approximately $500$ independent samples. For $N=12$, we used about $175$ samples, while for $N=14$ we averaged over $15$ independent circuit realizations.}} 
\end{figure*}

{
Before presenting a detailed spectral analysis of the Floquet circuits, we first summarize their generic behavior across several edges, namely, Id --- SWAP --- DCNOT --- CNOT --- Id. Through rigorous numerical analysis provided in Figs.~\ref{fig:magivsR} and ~\ref{fig:magivsR1}, we observe that quantum chaos is suppressed when at least one of the gate characteristics, $m_p$, $e_p$ and $g_t$, vanishes, even if the other two are considerably large \footnote[2]{A notable case is when the gates have vanishing $e_p$ and $g_t$, which occurs when they are locally equivalent to the Identity operation. The emergence of Poisson statistics has been rigorously proved for such Floquet circuits when the single-qubit gates are Haar random \cite{tkocz2012tensor}. Interestingly, in this case, $ m_p $ of individual gates can attain a nearly maximal value.}. In particular, we use Eq.~\eqref{epgt} in conjunction with the results from Appendix~\ref{app-b} to understand the emergence of quantum chaos from the interplay of these gate characteristics. For example, two-qubit gates located at the vertices of the tetrahedron, with local operations restricted to single-qubit Clifford gates, have vanishing non-stabilizing power. Spectral statistics for these special cases indicate that quantum chaos is fully suppressed. Quantum chaos can emerge as the gates are tuned away from the vertices when all three gate characteristics are simultaneously non-zero. However, when at least two of the characteristics are close to zero, we observe a critical transition from regular to chaotic behavior [see Figs.~\ref{fig:magivsR}(a) and ~\ref{fig:magivsR1}]. We observe this behavior along the edges Id --- SWAP and Id --- CNOT. In the following, we discuss these features in detail for all the aforementioned edges. 

}

\subsection{Edge Id --- SWAP}
\label{r_idswap}

We first consider the edge connecting the gates that are equivalent (up to single-qubit Clifford gates) to Identity and SWAP operations. Along this edge, the two-qubit gates are parametrized with a single parameter $c_x=c_y=c_z=J$, where  $J$ varies from $0$ to $\pi/2$. {The end points at $J=0$ and $\pi/2$ correspond to Identity and SWAP, i.e., Clifford gates.} Consequently in these two limits, the circuit in Eq.~\eqref{eq:Floquet_unitary} becomes an $N$-qubit Clifford operator and displays regular behavior with many fold degeneracies in the eigenphases. To explore the quantum chaotic nature away from these two points, we vary $J$ within two distinct ranges, $0+\epsilon\leq J\leq\pi/4$ and $\pi/4\leq J\leq\pi/2-\epsilon$, respectively, where we choose $\epsilon = 0.001$. At $J=\pi/4$, the gates are locally equivalent to $\sqrt{\text{SWAP}}$. Considering the above two ranges of $J$ separately allows us to reveal an interesting regular to chaos transition as detailed below.

The trend of $\langle r \rangle$ for the considered ranges of $J$ is presented in the upper panels of Figs.\ref{fig:magivsR}(a) and \ref{fig:magivsR}(b), respectively. 
As it can be seen from Fig.~\ref{fig:magivsR}(a), the curves for different system sizes cross near $J\approx 0.26$, indicating a transition from the regular to the quantum-chaotic regime. Through finite-size scaling analysis, we obtain the critical exponents $\nu \approx 0.65$ and $\zeta \approx 0.0014$, yielding an excellent collapse of the data for different system sizes [inset in Fig.~\ref{fig:magivsR}(a)].  
 
At the endpoint of this edge, we notice a chaotic to regular transition of the Floquet circuit as $J$ approaches $\pi/2$ [see Fig.~\ref{fig:magivsR}(b)]. However, the size dependence of the data suggests that in the thermodynamic limit, the transition will occur sharply at $J=\pi/2$. It is also worthwhile to note that $\langle r\rangle$ attains the CUE value rapidly with increasing system size in the vicinity of $J = \pi/2$. This behavior is illustrated in the inset of Fig.~\ref{fig:magivsR}(b), where we plot $\langle r\rangle_{\text{CUE}}-\langle r\rangle$ versus $1/N$ for $J=\pi/2-0.01$ on a semi-log scale. We observe that $\langle r\rangle$ approaches $\langle r\rangle_{\text{CUE}}$ exponentially with increasing $N$.

We can compare these findings with the behavior of $m_p$, $e_p$, and $g_t$ of the two-qubit gates, shown in the lower panels of Fig.~\ref{fig:magivsR} (also see Appendix~\ref{app-b}).  
Along this edge, $e_p$ and $g_t$ vary as $2\sin^2(J)\cos^2(J)$ and $\sin^2(J)$, respectively \cite{jonnadula2020entanglement}. While $g_t$ increases monotonically and reaches its maximum at $J = \pi/2$,  $e_p$ and $m_p$ grow upto $J=\pi/4$, after which they decrease and vanish at $J=\pi/2$. We can make two observations: First, when some of $m_p$, $e_p$, and $g_t$ vanish exactly, the Floquet level-spacing statistics indicates regular behavior. Second, the extended regular regime appears in a parameter space where all three of $m_p$, $e_p$, and $g_t$ are small. In contrast, when one of these quantities is large and the others are non-zero, the level-spacing statistics corresponds to a chaotic circuit. In what follows, we analyze the other edges, finding analogous behavior.

\subsection{Edges SWAP --- DCNOT --- CNOT }

In Fig.~\ref{fig:magivsR}(c), we examine Floquet circuits comprising the gates along the SWAP--DCNOT edge, parameterized by $c_x = c_y = \pi/2$ and a variable $c_z\in [\epsilon,\pi/2-\epsilon]$. 
For $c_z \lesssim \pi/2$, where the gates approach SWAP, the behavior of $\langle r \rangle$ is very similar to the approach to the end of the Id--SWAP edge analysed above. Accordingly, we expect a sharp transition in the thermodynamic limit from chaotic to regular dynamics at $c_z = \pi/2$ as observed in Fig.~\ref{fig:magivsR}(b). Likewise, at small $c_z$ where the gates approach DCNOT, our data indicates a sharp transition from chaos to regularity, as highlighted in the inset of Fig.~\ref{fig:magivsR}(c). For a fixed system size, convergence to CUE statistics is faster near the DCNOT gate than near the SWAP gate. {One can observe a similar behavior when the gates are taken from the DCNOT–CNOT edge with the parametrization $(\pi/2, c_y, 0)$, as illustrated in Fig.~\ref{fig:magivsR}(d).}

It is again instructive to compare the behavior of {$\langle r\rangle$} to that of $m_p$, $e_p$, and $g_t$. {The non-stabilizing power along these two edges varies as $m_p=\sin^2(2c_i)/5$ with $c_i=c_z$ for SWAP-DCNOT and $c_i=c_y$ for DCNOT-CNOT [see Eqs.~\eqref{eq:mp_swap-dcnot} and ~\eqref{eq:mag_cnot_dcnot}]. On the other hand, $e_p$ and $g_t$ vary in accordance with Eq.~\eqref{epgt}} --- $e_p$ vanishes only at the SWAP vertex, while $g_t$ remains non-zero throughout the edges. Therefore, excepting the vertices, all three gate characteristics are non-zero and at least one of them (in particular $g_t$) remains significantly large. This behavior is consistent with Fig.~\ref{fig:magivsR}(b), where quantum chaos emerges when $m_p$, $e_p$, and $g_t$ are all significantly large, reinforcing their connection to chaotic dynamics. {Moreover, the circuits near DCNOT consistently display stronger chaotic signatures than those near Id, SWAP, and CNOT.} Our results suggest that if more of these quantities vanish, the finite-size tendency towards regular behavior is more pronounced.




\subsection{Id --- CNOT edge} 
\begin{figure}[h!]
\includegraphics[scale=0.475]{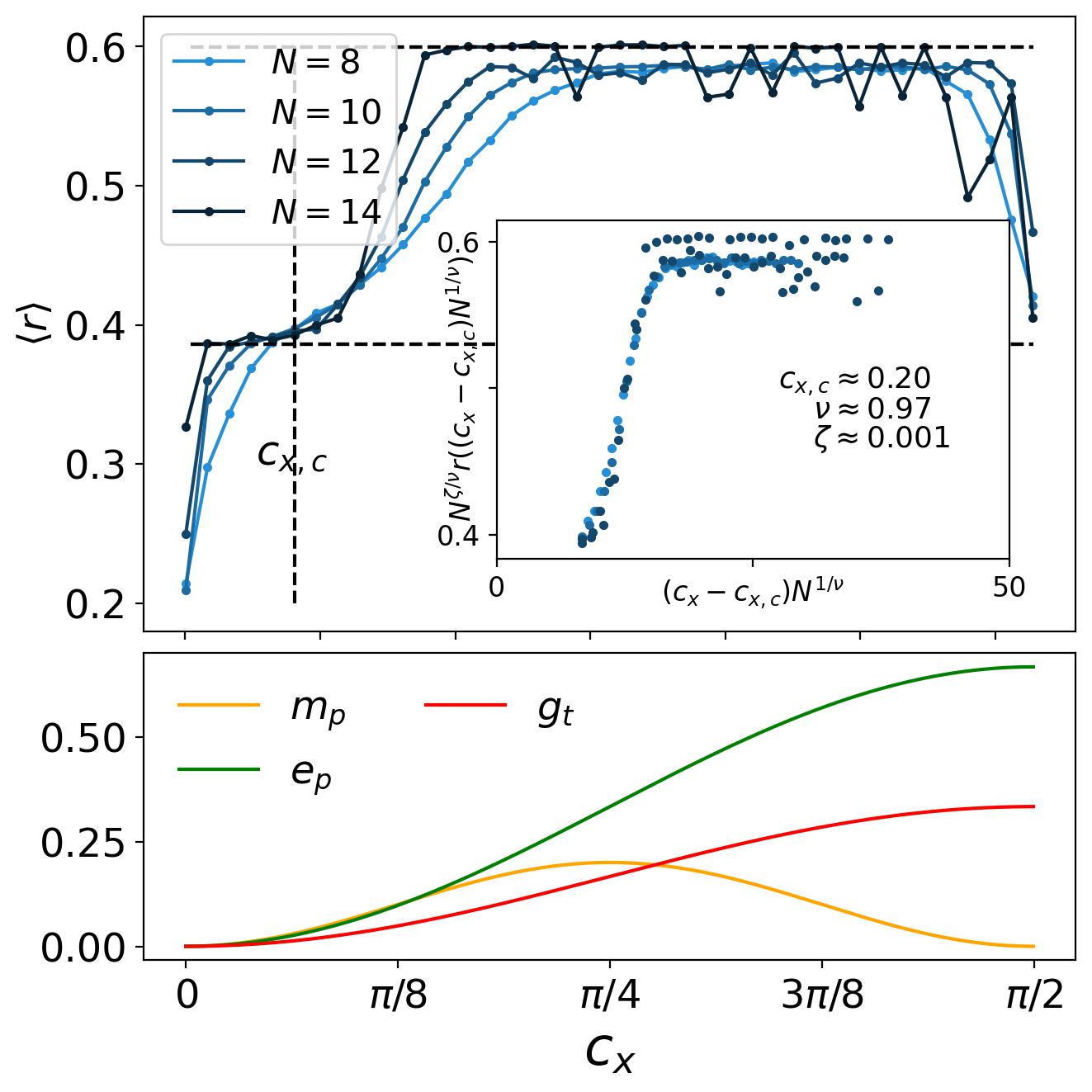}
\caption{\label{fig:magivsR1} Emergence of quantum chaos in the Floquet circuit comprising three brick-wall layers of the two-qubit gates along the edge Id --- CNOT, illustrated using $\langle r\rangle$ for different system sizes. Here, Euler angle $c_x$ is varied from $0+\epsilon$ to $\pi/2-\epsilon$, where $\epsilon =0.001$, while $c_y=c_z=0$. Near $c_x\approx 0.20$, curves corresponding to different system sizes cross, indicating a critical transition from regular to quantum chaotic regime. Inset: finite-size scaling analysis and data collapse for the critical transition. The observed critical exponents are $\nu\approx 0.97$ and $\zeta\approx 0.001$. Towards the endpoint at $c_x=\pi/2$, the chaotic behavior of the Floquet circuits at finite system sizes is suppressed. The lower panel demonstrates the behaviour of the quantities $m_p$, $e_p$, and $g_t$. At the endpoint $c_x=0$, all three quantities vanish, where it is known that quantum chaos is suppressed. As long as all the quantities remain small, the system remains in a regular phase, and transitions into a chaotic regime only once they are sufficiently large.  
This corroborates our observation in the case of the Id --- SWAP edge. Conversely, at the other endpoint $c_x=\pi/2$, $e_p$ attains its maximum value, $g_t$ remains considerably large, and $m_p$ vanishes. Quantum chaos is suppressed at finite system sizes, but the data suggests regular behavior to survive only at exactly the endpoint in the thermodynamic limit. For $N=8$ and $10$, we sample approximately $10^3$ circuit realizations, whereas for $N=12$ and $14$, we sample over approximately $125$ and $10$ realizations, respectively. } 
\end{figure}

Here, we carry out the spectral analysis of Floquet circuits that have gates from the edge connecting the identity to the CNOT operation. Along this edge, one of the three Euler angles, let us say $ c_x $, varies from $ 0 $ to $ \pi/2 $, while the other two parameters remain fixed at zero. Similar to the previous cases, the local operations are randomly chosen single-qubit Clifford gates. To probe quantum chaos, we compute $ \langle r \rangle $, averaged over many realizations of the Floquet operator as $ c_x $ is slowly varied. In contrast to the above edges, we observe that single-step Floquet circuits, defined in Eq.~(\ref{eq:Floquet_unitary}), exhibit many-fold degeneracies in their quasi-eigenspectra. However, $\langle r\rangle$ is not properly defined for degenerate spectra. To circumvent this issue, we introduce multiple brick-wall layers within a single time step. This adjustment ensures that the degeneracies are lifted and allows the dynamics to accurately predict the behavior associated with the considered gates. The corresponding numerical results are presented in Fig.~\ref{fig:magivsR1}.

In the figure, $\langle r\rangle$ is plotted against $c_x$ for the same system sizes as in the previous cases. 
In the numerical simulations, three consecutive layers of Eq.~(\ref{eq:Floquet_unitary}) are treated as a single time step. For small $c_x$, $\langle r\rangle$ starts below the value $\langle r\rangle_{\text{Poisson}}\approx0.386294$ due to nearly degenerate spectra. As $c_x$ is tuned away from zero, the curves corresponding to different system sizes intersect near $c_{x, c} \approx 0.20$, signaling a critical transition from regular to quantum chaotic behavior. The inset in Fig.~\ref{fig:magivsR1} shows the finite size scaling analysis and the corresponding data collapse for various system sizes, along with the extracted critical exponents $\nu = 0.97$ and $\zeta = 0.001$. Beyond the critical point, $\langle r\rangle$ remains close to the CUE value. However, as $c_x$ approaches the endpoint $\pi/2$, the spacing statistics begin to deviate from the CUE prediction, indicating a suppression of quantum chaos at the endpoint.

In line with our findings for the other edges, we observe a strong correlation between $\langle r\rangle$ and the gate properties $m_p$, $e_p$, and $g_t$. Along this edge, these quantities take the forms $m_p = \sin^2(2c_x)/5$, $e_p = 2\sin^2(c_x)/3$, and $g_t = \sin^2(c_x)/3$. At $c_x = 0$, the circuit is completely separable and exhibits no chaotic behavior. However, near the critical point $c_x \approx 0.3046$, all of these quantities remain positive but very small. This pattern corroborates our earlier observations in Sec.~\ref{r_idswap}, where we noticed a similar transition when all the gate invariants remain very small. Beyond the critical point, $e_p$ and $g_t$ increase monotonically, with $e_p$ approaching its maximal value of $2/3$. In contrast, $m_p$ exhibits a symmetric profile around $c_x = \pi/4$ and vanishes at $c_x = \pi/2$. The numerical results indicate that chaotic behavior is gradually suppressed (in finite system sizes) as $c_x$ approaches $\pi/2$. These findings are fully consistent with those from other edges: in all considered cases, regular dynamics emerges only when either any one of $m_p$, $e_p$, or $g_t$ vanishes exactly or if all three are small. Conversely, if all three are non-vanishing and any of these is sufficiently large, the system is in a quantum chaotic regime.

\section{Summary and Discussion}
\label{sec-conclusion}

In summary, we have investigated (i) non-stabilizing properties of quantum circuits consisting of interlaced Clifford and non-Clifford operations, and (ii) the emergence of quantum chaos in brick-wall Floquet circuits due to the interplay of non-stabilizing power ($m_p$), entangling power ($e_p$), and gate-typicality ($g_t$). To lay the groundwork to address these two aspects, we have analyzed the non-stabilizing properties of two-qubit gates using their standard parametrization via the Cartan decomposition \cite{khaneja2001time, kraus2001optimal}. Under local unitary equivalence, these gates are confined within a tetrahedron geometry known as the Weyl chamber and parametrized by three Euler angles. Vertices of the tetrahedron correspond, up to single-qubit unitaries, to the two-qubit Clifford gates Identity, CNOT, DCNOT, and SWAP. We have focused on gates along the edges of this tetrahedron and considered the single-qubit unitaries to be Clifford operators. Through analytical arguments and numerical simulations, we have shown that the non-stabilizing power along these edges varies smoothly as a function of the corresponding Euler parameters.

The central focus of this work concerns the evolution of non-stabilizing power when non-Clifford operations are repeatedly applied in conjunction with random and independent Clifford unitaries. These scenarios appear frequently in quantum computing protocols such as randomized benchmarking \cite{magesan2012efficient}. By studying the setting in Fig.~\ref{fig:sch2}, we have pinpointed the exact role played by the random Clifford unitaries in driving the non-stabilizing dynamics. The corresponding result is provided in Theorem~\ref{theo1}. Perhaps surprisingly, this result tells us that $\langle m_p(VCU) \rangle_C$ depends only on $m_p(U)$ and $m_p(V)$ with $\overline{m_p}$ being the only fixed point. As a result, even an arbitrarily small but nonzero amount of non-stabilizerness in a gate $CU$ is sufficient to drive the system towards the typical non-stabilizing power of a Haar-random unitary. Repeating such a gate $t$ times, with independent random Clifford unitaries $C$ at each step, leads to exponential convergence toward $\overline{m_p}$. This behavior reflects the thermalization of non-stabilizing properties in generic quantum circuits. It also provides new insights into the emergence of chaos and the limitations of classical simulability for such circuits. Moreover, an arbitrarily small non-stabilizing unitary can be used to generate magic equivalent to that of T-gates, thereby enabling fault-tolerant quantum computation, as detailed in Appendix~\ref{app-F3}. In addition, Theorem~\ref{theo1} can be used to define the operator-space non-stabilizing power (OSNP), the average non-stabilizing power imprinted onto a random Clifford operation as it is evolved under an arbitrary non-Clifford unitary. We have demonstrated this quantity for unitary time-evolution generated by the Ising Hamiltonian with transverse and longitudinal fields. Interestingly, we have found no remarkable difference in the OSNP for Ising Hamiltonians chosen in the integrable or chaotic regime.

The latter part of this work explores how the non-stabilizing power, along with the entangling power and gate typicality of two-qubit gates, drives brick-wall Floquet circuits toward quantum chaos. 
To this end, we have constructed brick-wall circuits comprising two-qubit gates with fixed Euler angles and random single-qubit Clifford unitaries, thereby ensuring minimal randomness in the system. 
Circuits constructed from two-qubit gates connecting CNOT to DCNOT and DCNOT to SWAP exhibit an immediate onset of quantum chaos as soon as the gates are perturbed away from the limiting Clifford cases. Interestingly, along the edges Id --- SWAP and Id --- CNOT, we observe sharp transitions from regular to chaotic behavior at critical points considerably away from the endpoints.  
In all analyzed cases, quantum chaos remains fully suppressed when at least one of $m_p$, $e_p$, and $g_t$ vanishes or when all of them simultaneously remain close to zero. 
In contrast, when all of these three quantities are sufficiently large, quantum chaos emerges. 
It remains an open problem whether this observation is fully generic and if $m_p$, $e_p$, and $g_t$ contribute equally to the onset of quantum chaos.


A further relevant question for future research is whether the exponential thermalization persists in the presence of external noise, and whether other measures of non-stabilizerness display similar thermalization patterns.
Moreover, the analytical explanation behind the exponential relaxation, derived from Theorem~\ref{theo1}, relies on the fact that the Clifford unitaries are spatially independent. Instead, correlations between the interspersed Clifford unitaries would prevent the decoupling of different $ m_p$'s. It would be interesting to probe if such correlations could suppress the thermalization of $m_p$. Another interesting future direction is to test if the higher-order Renyi entropies follow a similar result as Eq.~(\ref{mainres}).  
The observed connection between non-stabilizerness, entanglement, and chaos also points to possible ways of tuning circuit behavior, for example, to control chaotic dynamics and randomness generation \cite{choi2023preparing}. These ideas may find relevance in benchmarking, error correction, and the development of hybrid quantum algorithms.

\section{Data Availability Statement}
All data supporting the findings of this study are publicly available in the Zenodo repository at \cite{varikuti2026clifford}. 

\begin{acknowledgments}

N.D.V. acknowledges useful discussions with Arul Lakshminarayan and Shraddhanjali Choudhury on entangling power and quantum chaos. 
This project has received funding from the Italian Ministry of University and Research (MUR) through project DYNAMITE QUANTERA2\_00056, in the frame of  ERANET COFUND QuantERA II – 2021 call co-funded by the European Union (H2020, GA No 101017733); the European Union - Next Generation EU, Mission 4, Component 2 - CUP E53D23002240006, and CARITRO through project SQuaSH. This project was supported by the European Union under Horizon Europe Programme - Grant Agreement 101080086 - NeQST; the Swiss State Secretariat for Education, Research and lnnovation (SERI) under contract number UeMO19-5.1; the Provincia Autonoma di Trento, and Q@TN, the joint lab between University of Trento, FBK—Fondazione Bruno Kessler, INFN—National Institute for Nuclear Physics, and CNR—National Research Council.
S.B.\ acknowledges CINECA for use of HPC resources
under Italian SuperComputing Resource Allocation– ISCRA
Class C Projects No. DISYK-HP10CGNZG9 and DeepSYK - HP10CAD1L3. 
Views and opinions expressed are however those of the author(s) only and do not necessarily reflect those of the European Union or of the Ministry of University and Research.
Neither the European Union nor the granting authority can be held responsible for them.

\end{acknowledgments}

\newpage
\onecolumngrid
\appendix

\section{Two-qubit gates and vertices of the Tetrahedron geometry }
\label{app-a}

{
In this appendix, we briefly outline the geometry of two-qubit gates and indicate the locations of Clifford gates within this geometry. As also mentioned in the main text, any two-qubit unitary operator $U\in \text{SU}(4)$ can be written in the canonical form using the Euler angles $\{c_x, c_y, c_z\}$ up to left and right multiplication of arbitrary single-qubit unitaries as
\begin{equation}\label{cartan}
 U=\exp\left\{-i\sum_{j\in\{x, y, z\}}\dfrac{c_j}{2}\left(\sigma_{j}\otimes\sigma_{j}\right)\right\},
 \end{equation}
where $c_j\in (0, \pi]$ and $\{\sigma_j\}_{j\in\{x, y, z\}}$ are the Pauli operators. Imposing local equivalence, i.e., requiring that two unitaries related by local transformations share the same set of Euler angles, restricts the ranges of the Euler angles to $ 0\leq c_z \leq c_y \leq c_x \leq \pi/2 $ \cite{zhang2003geometric}; see Fig.~\ref{fig:schem} of the main text. As we shall show below, the vertices of this geometry are locally equivalent to well-known two-qubit Clifford unitaries: Identity ($c_x,c_y,c_z=0, 0, 0$), CNOT ($\pi/2, 0, 0$), double-CNOT or \hbox{DCNOT} also known as iSWAP ($\pi/2$, $\pi/2$, 0), and SWAP ($\pi/2$, $\pi/2$, $\pi/2$) gates. Note that DCNOT and SWAP gates can be constructed using two and three CNOT gates, respectively, as also illustrated below.

}

We now demonstrate that the vertices of the tetrahedron in Fig.~\ref{fig:schem} correspond to Clifford unitaries, assuming the single-qubit gates involved are also Clifford operations. To establish this, it suffices to show that the corresponding two-qubit unitaries map Pauli strings to Pauli strings, up to overall phase factors. The vertex with all zeros is a trivial case, as it corresponds to the identity operation, up to local single-qubit unitaries.
We now consider the vertex with the Euler angles $(c_x, c_y, c_z)=(\pi/2, 0, 0)$. To show that this vertex corresponds to a Clifford operator, we write the corresponding unitary as 
\begin{eqnarray}
U_{(\pi/2, 0, 0)}=\exp\left\{ -i\dfrac{\pi}{4}\sigma_x\otimes\sigma_x \right\} =\dfrac{1}{\sqrt{2}}\left[ \mathbb{I}_{4}-i\left(\sigma_x\otimes\sigma_x\right) \right]
\end{eqnarray}
In the case of two qubits, there are a total of $16$ Pauli strings, including the identity operator. For an arbitrary Pauli string from this set $\sigma_a\otimes\sigma_b$, the action of $U$ is given by \begin{eqnarray}
U^{\dagger}_{(\pi/2, 0, 0)}\left( \sigma_a\otimes\sigma_b \right)U_{(\pi/2, 0, 0)}=\dfrac{1}{2}\left[ \sigma_a\otimes\sigma_b -i\left[ \sigma_a\otimes\sigma_b, \sigma_x\otimes\sigma_x \right]+\left(\sigma_x\sigma_a\sigma_x\right)\otimes\left( \sigma_x\sigma_b\sigma_x \right) \right].  
\end{eqnarray}
Whenever $a=b$, it is known that $[\sigma_a\otimes\sigma_a, \sigma_x\otimes\sigma_x]=\mathbf{0}$ for all $a\in \{x, y, z\}$. Therefore, $U^{\dagger}_{(\pi/2, 0, 0)}(\sigma_a\otimes\sigma_a)U_{(\pi/2, 0, 0)}=\sigma_a\otimes\sigma_a$. In addition, the strings from the set $\{\sigma_x\otimes\mathbb{I}_{2}, \mathbb{Id}_{2}\otimes\sigma_x, \sigma_y\otimes\sigma_z, \sigma_z\otimes\sigma_y\}$ also commute with $U$. This leaves us with the set of strings $\{ \sigma_y\otimes\mathbb{I}_2, \sigma_z\otimes\mathbb{I}_2, \sigma_x\otimes\sigma_y, \sigma_x\otimes\sigma_z \}$ and their permutations, and we only need to check if these remaining eight strings map among themselves. By explicit calculation, we see that under the mapping of $U$, this set maps back to itself up to phases: 
\begin{eqnarray}
U^{\dagger}_{(\pi/2, 0, 0)}\left( \sigma_y\otimes\mathbb{I}_{2} \right) U_{(\pi/2, 0, 0)}&=&-\left(\sigma_z\otimes\sigma_x\right)\,,\nonumber\\
U^{\dagger}_{(\pi/2, 0, 0)}\left( \sigma_z\otimes\mathbb{I}_{2} \right) U_{(\pi/2, 0, 0)}&=&\sigma_y\otimes\sigma_x\,,\nonumber\\
U^{\dagger}_{(\pi/2, 0, 0)}\left( \sigma_x\otimes\sigma_y \right)U_{(\pi/2, 0, 0)}&=&-(\mathbb{I}_{2}\otimes\sigma_z)\,,\nonumber\\
U^{\dagger}_{(\pi/2, 0, 0)}\left(\sigma_x\otimes\sigma_z\right)U_{(\pi/2, 0, 0)}&=&\mathbb{I}_{2}\otimes\sigma_y\,.
\end{eqnarray}
Therefore, $U_{(\pi/2, 0, 0)}$ is a Clifford operator. Moreover, it is known that $\text{CNOT}=\exp\{i\frac{\pi}{4}((\mathbb{I}-\sigma_z)\otimes\sigma_x)\}$, which is equivalent to $U_{(\pi/2, 0, 0)}$ up to local operations. Hence, we denote $U_{(\pi/2, 0, 0)}=(\text{CNOT})_{\text{local}}$. 

We now consider the other two vertices and verify that the corresponding unitaries are again Clifford unitaries. At the vertex ($\pi/2, \pi/2, 0$), the unitary can be written as 
\begin{eqnarray}
U_{(\pi/2, \pi/2, 0)}&=&\exp\left\{ -i\dfrac{\pi}{4}\left( \sigma_x\otimes\sigma_x +\sigma_y\otimes\sigma_y \right) \right\}\nonumber\\
&=&\text{(CNOT)}_{\text{local}}\exp\left\{ -i\dfrac{\pi}{4}\sigma_y\otimes\sigma_y \right\}\nonumber\\
&=&(\text{CNOT})_{\text{local}}\left( S^{\dagger}\otimes S^{\dagger} \right)\exp\left\{ -i\dfrac{\pi}{4}\sigma_x\otimes\sigma_x \right\}\left( S\otimes S \right)\nonumber\\
&=&(\text{CNOT})_{\text{local}}\left( S^{\dagger}\otimes S^{\dagger} \right)(\text{CNOT})_{\text{local}} \left( S\otimes S \right), 
\end{eqnarray}
where $S$ is the phase gate. Since the final expression is a product of two locally equivalent CNOTs, $U_{(\pi/2, \pi/2, 0)}$ is again a Clifford operator. In the same way, one can show that the unitary corresponding to $(\pi/2, \pi/2, \pi/2)$ is also a Clifford operation and is equivalent to the SWAP gate up to a constant phase:
\begin{eqnarray}
U_{(\pi/2, \pi/2, \pi/2)}&=&(\text{CNOT})_{\text{local}}\left( S^{\dagger}\otimes S^{\dagger} \right)(\text{CNOT})_{\text{local}} \left( S\otimes S \right) \left( H^{\dagger}\otimes H^{\dagger} \right)(\text{CNOT})_{\text{local}} \left( H\otimes H \right)\nonumber\\
&=& -i\exp\left\{i\dfrac{\pi}{4}\right\}\text{SWAP}. 
\end{eqnarray}
If the Euler angles are not at the vertices of the tetrahedron, the corresponding gates are not locally equivalent to any Clifford operation. 
Therefore, although not necessary, any finite yet non-maximal $e_p$ is sufficient to have a nonzero $ m_p$ in a two-qubit gate.

{
\section{$m_{p}(U)$ of two-qubit gates up to single-qubit Clifford operations} 
\label{app-b}

In this appendix, we analyze the non-stabilizing power of two-qubit unitary gates and lay the groundwork for several results discussed in the main text. To facilitate this analysis, we consider the standard parametrization of two-qubit gates using the Cartan decomposition \cite{bertini2019exact}.

An extensive characterization of the entangling power for these gates using local invariants has been presented in Ref.~\cite{jonnadula2020entanglement}. Here, we provide a detailed analysis of $m_{p}$ for the gates in Eq.~(\ref{cartan}) up to single-qubit Clifford gates. Although these two powers are not complementary in a strict sense required to form a well-defined phase space \footnote[3]{For instance, the gate typicality is complementary to the entangling power in this sense, see Refs.~\cite{jonnadula2020entanglement, jonnadula2017impact}}, understanding their interplay can provide insights into the roles these quantities assume in the context of quantum chaos \cite{lakshminarayan2001entangling, pal2018entangling, styliaris2021information, varikuti2022out, leone2021quantum}, classical simulability, quantum error correction \cite{scott2004multipartite, rather2022thirty}, and fault-tolerant quantum computation \cite{howard2017application}. To proceed, we first consider the local operations in Eq.~(\ref{cartan}) to be single-qubit identity operators and study $m_{p}$ as a function of the Euler angles. In particular, we consider the gates that reside along the edges of the tetrahedron, denoted by Id --- CNOT, CNOT --- DCNOT, DCNOT --- SWAP, and SWAP --- Id. It is worth mentioning that these edges form the boundary of two-qubit gates in the space of $e_p$ and $g_t$ \cite{jonnadula2020entanglement}. In the following, we shall calculate $m_{p}$ along these edges. 

\textbf{\textit{Along} Id --- CNOT \textit{edge:}} 
For $c_x=c_y=c_z=0$, the unitary is locally equivalent to the identity operation, whereas it matches the CNOT gate if one of these parameters takes the value $\pi/2$. Hence, without restriction of generality, to examine $m_{p}(U)$ along this edge, we consider $c_y=c_z=0$ and vary $c_x$, yielding the unitary to be examined as $U=\exp\{-ic_x\sigma_x\otimes\sigma_x/2\}$. We now proceed to derive an expression for $m_{p}(U)$ of this unitary as a function of the Euler angle $c_x$. 

As per Eq.~\eqref{nonst2}, we need to average over the complete set of stabilizer states to calculate $m_{p}(U)$. There are a total of sixty stabilizer states in the two-qubit Hilbert space \cite{gross2006hudson}. One can easily verify that twelve among these are eigenstates of the above unitary $U$, implying that the action of $U$ does not take these states out of the stabilizer space. For any of the remaining $48$ states, we find that the set comprising the expectations of Pauli strings $\{\langle\psi |U^{\dagger} P_i U|\psi\rangle^2\}_{i=0}^{15}$ contains exactly six non-zero entries. While it depends on the stabilizer state for which Pauli strings $P_i$ the non-zero expectation values occur, they always take values from the set $\mathcal{F}\equiv \{1, \sin^2\left( c_x\right), \cos^2\left( c_x \right)\}$ with each value appearing twice. Consequently, the amount of non-stabilizerness generated by the action of $U$ on each of these stabilizer states is
\begin{eqnarray}
\mathcal{M}(U|\psi\rangle)&=&1-2^2\sum_{i=0}^{15}\dfrac{1}{2^4}\langle\psi |U^{\dagger}P_iU|\psi\rangle^4 \nonumber\\
&=&1-\dfrac{1}{4}\left( 2+2\sin^4(c_x)+2\cos^4(c_x) \right)\nonumber\\
&=&\sin^2\left(c_x\right)\cos^2\left( c_x \right). 
\end{eqnarray}
Thus, the non-stabilizing power of $U$ becomes
\begin{equation}\label{nonstab_exp}
m_{p}(U)=\dfrac{48}{60}\sin^2\left( c_x \right) \cos^2\left( c_x \right)
=\dfrac{\sin^2\left( 2c_x \right)}{5}, 
\end{equation}
Also, since the non-stabilizerness is invariant under arbitrary local Clifford transformations, for any $U'=\exp\{-iJ\sigma_a\sigma_b/2\}$, where $a, b\in\{x, y, z\}$, it follows that 
\begin{eqnarray}
m_{p}(U')=\dfrac{\sin^2(2J)}{5}. 
\end{eqnarray}

In Fig.~\ref{fig:nonst1}(a), for completeness, we present the numerically computed $m_p(U)$ as $c_x$ is varied from $0$ to $\pi/2$. In these simulations, we fix the single-qubit unitaries to be identities.  Since the set of two-qubit stabilizer states is finite, as also discussed above, $m_{p}$ can be computed by numerically summing over all the states of the form $U|\psi\rangle$. {As anticipated, the numerics show an exact agreement with the analytical expression derived in Eq.~(\ref{nonstab_exp})}.

\begin{figure}[h]
\includegraphics[scale=0.485]{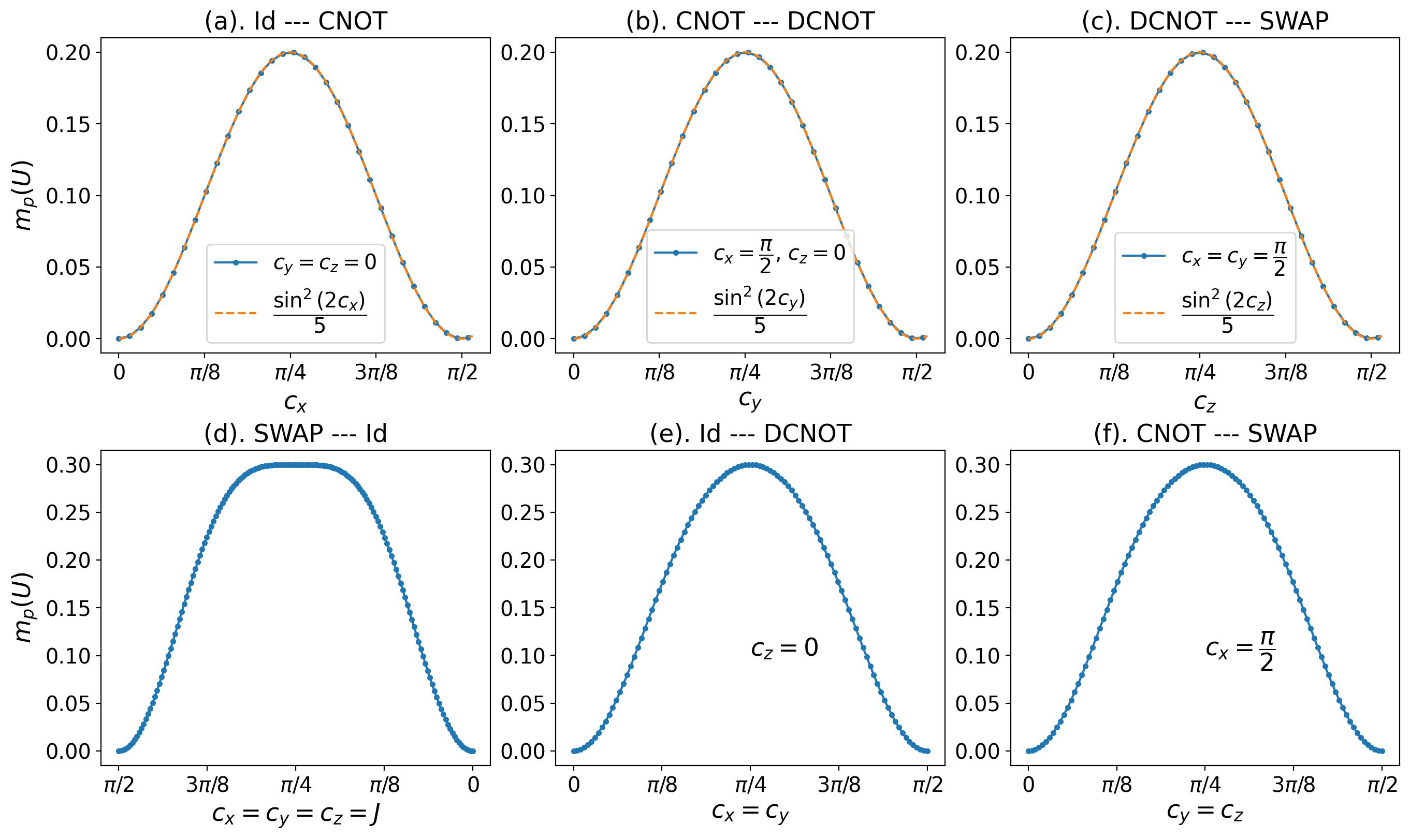}
\caption{\label{fig:nonst1} Non-stabilizing power of the two-qubit unitaries as the Euler angles are varied while the local unitaries are taken to be Cliffords. (a) Id --- CNOT ($c_x = c_y = 0$, $c_z$ is varied from $0$ to $\pi/2$). Note that the non-stabilizing power remains invariant if $c_z$ is replaced by either $c_x$ or $c_y$, provided the other two parameters are kept at zero. 
(b) CNOT --- DCNOT ($c_x = \pi/2$, $c_z = 0$, $c_y$ is varied). Along this edge, the two-qubit gates have maximal entangling power. 
(c) DCNOT --- SWAP ($c_x=c_y=\pi/2$, variable $c_z$). 
(d) Fractional powers of the SWAP operator, $-i\text{SWAP}^{2J/\pi}$ for $J\in [0, \pi/2]$. 
(e) Id --- DCNOT ($c_x=c_y$ is varied while $c_z$ is fixed at $0$).
(f) CNOT --- SWAP ($c_x=\pi/2$ and $c_y=c_z$ is varied).
In all panels, the numerical results for the non-stabilizing power are obtained by averaging over all sixty stabilizer states in the two-qubit space.
In panels (a-c), the numerical results match exactly with the analytical expectation $\sin^2(2c_j)/5$, with $c_j$ representing the parameter being varied. 
}  
\end{figure}

\textbf{\textit{Along} CNOT --- DCNOT \textit{edge}}: 
Up to single-qubit Clifford operations, the unitaries along this edge can be parametrized as
\begin{eqnarray}\label{CDC}
U&=& \exp\left\{ -i\left(\frac{\pi}{4}\sigma_x\otimes\sigma_x +\frac{c_y}{2}\sigma_y\otimes\sigma_y \right)\right\}\nonumber\\
&=& \left(\text{CNOT}\right)_{\text{loc}} \exp\left\{ -i\dfrac{c_y}{2}\sigma_y\otimes\sigma_y \right\}.
\end{eqnarray}
Here, $(\text{CNOT})_{\text{loc}}$ indicates that the unitary $\exp\{-i\pi\sigma_x\otimes\sigma_x/4\}$ is equivalent to the CNOT gate up to local Clifford operations. 
The interested reader may find further details in Appendix \ref{app-a}. Since the right-hand side of Eq.~\eqref{CDC} corresponds to a Clifford transformation of the unitaries along the Id --- CNOT edge, the non-stabilizing power can be equivalently expressed as
\begin{eqnarray}\label{eq:mag_cnot_dcnot}
m_{p}(U)=\dfrac{1}{5}\sin^2(2c_y).  
\end{eqnarray}
We numerically verify this result in Fig.~\ref{fig:nonst1}(b), where we plot $m_{p}$ against the Euler angle $c_y$. The numerical results show an exact agreement with Eq.~\eqref{eq:mag_cnot_dcnot}. Notably, when $c_y=\pi/2$, the non-stabilizing power vanishes, implying that the unitary is a Clifford operator. The corresponding unitary is locally equivalent to the DCNOT gate. 

\textbf{\textit{Along} DCNOT --- SWAP \textit{edge}}: 
The arguments presented above can be extended to calculate the $m_p(U)$ along the edge DCNOT --- SWAP, where the unitary can be expressed as 
\begin{eqnarray}\label{dcnot}
U&=&\exp\left\{ -i\left(\frac{\pi}{4}\sigma_x\otimes\sigma_x +\frac{\pi}{4}\sigma_y\otimes\sigma_y+\frac{c_z}{2}\sigma_z\otimes\sigma_z \right)\right\}\nonumber\\ 
&=&(\text{DCNOT})_{\text{loc}} \exp\left\{ -i\dfrac{c_z}{2}\sigma_z\otimes\sigma_z \right\}. 
\end{eqnarray}
Therefore, like in the previous case, the non-stabilizing power can be written as
\begin{eqnarray}
m_{p}(U)=\dfrac{1}{5}\sin^2(2c_z).  
\label{eq:mp_swap-dcnot}
\end{eqnarray}
The exact agreement between numerically obtained $m_{P}(U)$ with the above expression is presented in Fig.~\ref{fig:nonst1}(c).

\textbf{\textit{Along} Id --- SWAP \textit{edge}}:  The unitaries along this edge can be parametrized using $J\in [0, \pi/2)$ as 
\begin{eqnarray}
U&=&\exp\left\{ -i\dfrac{J}{2}\sum_{j\in\{x, y, z\}}\sigma_j\otimes\sigma_j \right\}\nonumber\\
&=&e^{iJ/2}\exp\{-iJ (\text{SWAP})\}\nonumber\\
&=&e^{iJ/2}\{-i(\text{SWAP})\}^{2J/\pi}, 
\label{eq:id-swap-param}
\end{eqnarray}
where the last equality follows from the known relation $\exp\{-i\pi(\text{SWAP})/2\}=-i(\text{SWAP})$. Thus, these unitaries are equivalent to fractional powers of the SWAP operator. Unlike the previous cases as presented in Eqs.~\eqref{CDC} and \eqref{dcnot}, fractional powers of the SWAP operator cannot be connected to the unitaries along the remaining edges via Clifford transformations. Consequently, the expression for the non-stabilizing power can not be derived in the same manner. For this case, we numerically evaluate $m_{p}$ as a function of the isotropic interaction strength $J$. The corresponding results are shown in Fig.~\ref{fig:nonst1}(d). Towards the endpoints $J=0$ and $J=\pi/2$, the non-stabilizing power vanishes, as expected. As $J$ moves away from these points, $m_p$ increases monotonically till $\pi/4$, around which point it is symmetric.  

\textbf{\textit{Along} Id --- DCNOT \textit{and} SWAP --- CNOT \textit{edges}:} Apart from the four edges discussed above, there are two additional edges of the tetrahedron, namely, Id---DCNOT and SWAP---CNOT. While we make use of the results for the above edges to illustrate our findings in the main text (see Figs.~\ref{fig:nonstab1}, ~\ref{fig:nonstab2}, ~\ref{fig:magivsR} and ~\ref{fig:magivsR1} of the main text), it is also of interest to study the non-stabilizing power of the gates located along these two edges. The unitaries along the edge Id---DCNOT can be parametrized using a single parameter as 
\begin{eqnarray}
  U=\exp\left\{ -i\left( \dfrac{J}{2}\sigma_x\otimes\sigma_x+\dfrac{J}{2}\sigma_y\otimes\sigma_y  \right) \right\}, \thickspace\text{ where } c_x=c_y=J, c_z=0 \,,
\end{eqnarray}
and the unitaries along the edge CNOT --- SWAP can be parametrized as  
\begin{eqnarray}
U'= \exp\left\{-i\left( \dfrac{\pi}{4}\sigma_x\otimes\sigma_x+\dfrac{J}{2}\sigma_y\otimes\sigma_y +\dfrac{J}{2}\sigma_z\otimes\sigma_z \right)\right\}, \thickspace\text{ where } c_x=\pi/2, c_y=c_z=J \,.    
\end{eqnarray}
The operator $U'$ can be transformed into $U$ through the action of Clifford operators. Therefore, it is natural to expect them to have similar non-stabilizing powers when the corresponding parameters $J$ are equal, which is indeed what we find --- see the numerical results shown in Fig.~\ref{fig:nonst1}(e) and ~\ref{fig:nonst1}(e). These curves do not follow the exact same law as that of the other edges shown , but they do share some qualitative similarities. Near the end points of the edges, $m_p$ vanishes as they represent Clifford gates. Moreover, the curves are symmetric around the midpoint of the edges. 

}

\section{$m_{p}(U)$ versus $e_p(U)$}
\label{app-c}

In this appendix, we visualize the classical simulability phase space of two-qubit gates by plotting $ m_p $ against $ e_p $. 
We first study the $e_p$--$m_p$ diagram of two-qubit gates for Euler angles $(c_x, c_y, c_z)$ randomly drawn from $[0, \pi/2]$ from the uniform distribution and randomly sampled single-qubit unitaries. The data corresponding to these random unitaries is shown in Fig.~\ref{fig:epvsmp} as blue bullets. 
It is instructive to compare these to analytically computable lines. 

We first argue that the gates from the edges Id --- CNOT and SWAP --- DCNOT constitute the lower boundary of the $ m_p $--$ e_p $ diagram. The non-stabilizing power of gates located along these edges has been discussed in Sec.~\ref{sec-2} of the main text. We begin by considering the family of gates lying along the Id --- CNOT edge, given by
\[
U(c_x) = \exp\left\{ -i \frac{c_x}{2} \sigma_x \otimes \sigma_x \right\},
\]  
(i.e., we choose only those gates along the Id --- CNOT edge where all single-qubit gates that can modify the Cartan decomposition have been set to identity operators). One can qualitatively argue that for a fixed value of $ c_x $,
$m_p\left( U(c_x) \right) $ is typically less than $ m_p\left( (v_1 \otimes v_2) U(c_x) (w_1 \otimes w_2) \right)$, for arbitrary non-Clifford single-qubit unitaries $ v_1, v_2, w_1, w_2 $. Recall that the Clifford group forms a unitary-3 design and thus constitutes an (over)complete basis in the operator space. As such, any unitary can be expanded in terms of Clifford operators. When the unitary is itself a Clifford operator, its expansion involves only a single term. For non-Clifford unitaries, the minimal number of terms in such an expansion can be as low as two — which is the case for $ U(c_x) $, as shown by  
\begin{equation} \label{expan-non}
U(c_x) = \cos\left( \frac{c_x}{2} \right)\mathbb{I}_4 - i \sin\left( \frac{c_x}{2} \right) \left( \sigma_x \otimes \sigma_x \right),
\end{equation}  
where both $ \mathbb{I}_4 $ and $ \sigma_x \otimes \sigma_x $ are Clifford operations. An arbitrary single-qubit unitary can be expressed as $u = \alpha_0 \mathbb{I}_2 + \alpha_x \sigma_x + \alpha_y \sigma_y + \alpha_z \sigma_z$, subject to the normalization condition $ \sum_i |\alpha_i|^2 = 1 $. Applying such local non-Clifford unitaries increases the number of terms in the expansion of Eq.~(\ref{expan-non}), pushing the overall unitary further away from the Clifford group and thus increasing its non-stabilizing power, while maintaining the entanglement-generating power. 
This behavior is particularly apparent when $ c_x = 0 $ or $ \pi/2 $, in which cases $ U(c_x) $ is locally equivalent (up to single-qubit Clifford gates) to the two-qubit Clifford gates Id, CNOT, SWAP, and DCNOT. These unitaries correspond to the vertices of the tetrahedron shown in Fig.~\ref{fig:schem}. When non-Clifford single-qubit unitaries are applied in these cases, the resulting two-qubit gate generically acquires non-zero non-stabilizing power. 
Thus, we expect the unitaries from the edges Id --- CNOT and SWAP --- DCNOT to have minimal $m_p$ at a given value of $e_p$. 

\begin{figure}[H]
\centering
\includegraphics[scale=0.5]{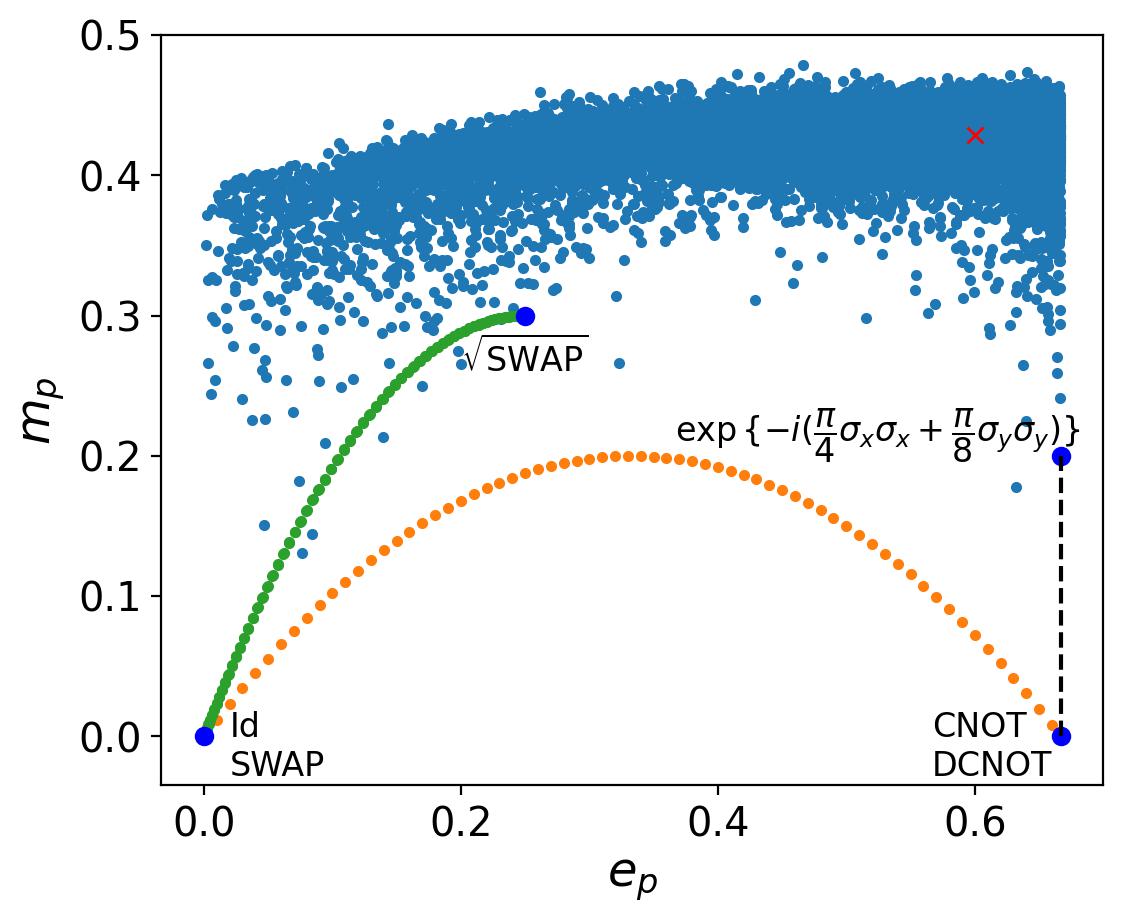}
\caption{\label{fig:epvsmp} Entangling power versus non-stabilizing power for samples of two-qubit unitaries drawn at random from the operator space. 
The red-colored marker at the point ($m_{p}, e_p$)=($1-4/7, 3/5$) indicates the Haar-average value. 
The orange curve is given by the parabolic equation $y=(6/5)x-(9/5)x^2$ and traces the values along the edges Id --- CNOT and SWAP --- DCNOT, yielding a lower bound on the data.  
The endpoints of these analytic curves, shown with blue dots, represent the two-qubit Clifford operations.
In addition, the vertical line towards the right represents the edge CNOT --- DCNOT and the green curve denotes the edge Id --- SWAP, whose highest points are reached at $c_x=\pi/2$, $c_y=\pi/4$, and $c_z=0$ and at $\sqrt{\mathrm{SWAP}}$, respectively. 
} 
\end{figure}

The values of non-stabilizing and entanglement generating power along these edges can be computed analytically for any unitary in the Cartan form or modifications thereof with single-qubit Clifford operations. 
The entangling power of two-qubit gates can be expressed using the Euler angles as \cite{jonnadula2020entanglement}
\begin{eqnarray}
e_p(U)=\dfrac{2}{3}\left[ \sin^2(c_x)\cos^2(c_y)+\sin^2(c_y)\cos^2(c_z)+\sin^2(c_z)\cos^2(c_x) \right]. 
\end{eqnarray}
Thus, along the edges Id --- CNOT and SWAP --- DCNOT the entangling powers are given by $e_p(c_x)=2\sin^2(c_x)/3$ and $e_p(c_z)=2\sin^2(c_z)/3$, respectively. 
The corresponding non-stabilizing powers are given by $m_p(c_x)=\sin^2(2c_x)/5$ and $m_{p}(c_z)=\sin^2(2c_z)/5$, respectively. 
By noticing that $\sin^2(2c_{x(z)})=3e_p(c_{x(z)})/2$, one can rewrite $m_{p}$ in terms of $e_{p}$ as
\begin{eqnarray}
m_{p}(c_{x(z)})&=&\dfrac{1}{5}\left[ 4\sin^2(2c_{x(z)})\left( 1-\sin^2(2c_{x(z)}) \right) \right] \nonumber\\
&=&\dfrac{6}{5}e_{p}(c_{x(z)})\left( 1-\dfrac{3}{2}e_{p}(c_{x(z)}) \right). 
\end{eqnarray}
The resulting curve in the $m_p$--$e_p$ plane is plotted in Fig.~\ref{fig:epvsmp} in orange color. 
All randomly sampled correlation data lies above this line, confirming the arguments for it to constitute a lower boundary. 


\section{Non-stabilizing power of two-qubit gates embedded in a four-qubit Hilbert space}
\label{app-d}

Here, we present numerical results for the non-stabilizing power of spatially extended two-qubit gates. Given a two-qubit gate $ U_2 $, we evaluate $ m_p(\mathbb{I}_{2} \otimes U_2 \otimes \mathbb{I}_{2}) $. The corresponding results are plotted in Fig.~\ref{fig:spaext}, where we focus on gates that lie along the same four edges of the tetrahedron in Fig.~\ref{fig:schem} as considered in the main text. Numerical data for randomly sampled unitaries are shown in blue.

The edges Id --- CNOT, CNOT --- DCNOT, and DCNOT --- SWAP are related to each other via Clifford transformations. In agreement with this observation, the obtained data is very similar. We also observe that the numerical data fit well to the function $ \sin^2(2c_i)/4.25 $, where $i=x,y,z$ along all three edges. The behavior along the edge Id --- SWAP is quantitatively different but appears to follow a similar pattern, with a maximum in the center and roughly symmetric behavior with respect to the control parameter.  

\begin{figure}[H]
    \centering
    \includegraphics[width=0.6\linewidth]{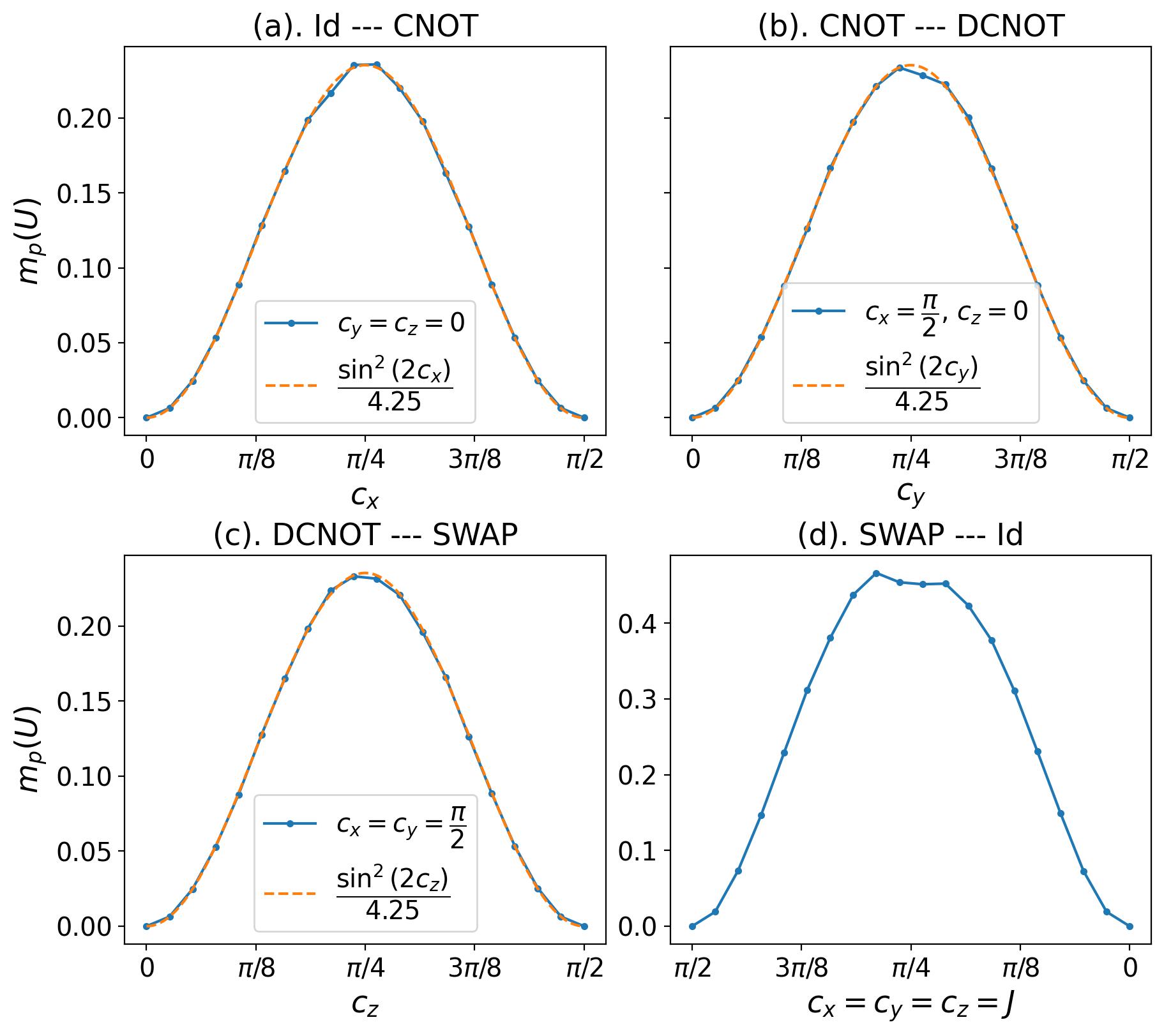}
    \caption{\label{fig:N4magic} Non-stabilizing power of two-qubit gates embedded in a four-qubit Hilbert space, i.e., of the form $\mathbb{I}_{2}\otimes U_2\otimes\mathbb{I}_{2}$. Here, $U_2$ is a two-qubit gate that lies along four edges of the tetrahedron in Fig.~\ref{fig:schem}. The results are averaged over nearly $\sim 10^3$ samples of stabilizer states supported over four qubits. Blue curves denote the numerical results. Orange curves in the first three panels indicate the analytical fit $\sin^2(2c_i)/4.25$, where $c_i$ represents the Euler angle that is being varied. }
    \label{fig:spaext}
\end{figure}

\section{Proof of Theorem \ref{theo1}}\label{app-e}
\begin{theorem*} (\textup{Restatement of Theorem~\ref{theo1}})
Let $U$ and $V$ be two arbitrary non-Clifford unitary operators supported over an $N$-qubit Hilbert space $\mathcal{H}=\mathbb{C}^{2^N}$ with non-stabilizing powers $m_p(U)$ and $m_p(V)$, respectively, and let $C$ be a Clifford operator sampled at random from the Clifford group according to its Haar measure. Then, the following relation holds:
\begin{equation}
\langle m_{p}(VCU) \rangle_{C} = m_{p}(U)+m_{p}(V)-\dfrac{m_{p}(U)m_{p}(V)}{\overline{m_{p}}},  
\end{equation}
where $\overline{m_{p}}=\langle m_{p}(W)\rangle_{W}$ denotes the non-stabilizing power averaged over the Haar-random unitaries $W\in \mathcal{U}(2^N)$.  
\end{theorem*}    
\begin{proof}
We are interested in evaluating $\langle m_{p}(VCU)\rangle_{C}$, the average non-stabilizing power when two non-Clifford operations $U$ and $V$ are interspersed with a random Clifford operation $C$. Using the definition of the non-stabilizing power given in Eq.~\eqref{eq:nonstabpowerdef}, we have
\begin{eqnarray}\label{cliffordintm1}
\langle m_{p}(VCU)\rangle_{C}&=&1-2^N\int_{C}d\mu(C)\text{Tr}\left[ V^{\dagger \otimes 4}QV^{\otimes 4} C^{\otimes 4}U^{\otimes 4}\overline{\left(|\psi\rangle\langle\psi |\right)^{\otimes 4}} U^{\dagger \otimes 4}C^{\dagger \otimes 4} \right], 
\end{eqnarray}
where $d\mu(C)$ is the invariant Haar measure associated with the Clifford group supported over a Hilbert space consisting of $N$ qubits ($\mathcal{H}=\mathbb{C}^{2^N}$).
Recall that $ \overline{\left( |\psi\rangle\langle\psi| \right)^{\otimes 4}} $ represents the fourth moment of the ensemble of stabilizer states in $\mathcal{H}=\mathbb{C}^{2^N}$, which can be obtained by evaluating the integral $ \int_{C} d\mu(C) \left( C |\psi\rangle\langle\psi | C^{\dagger} \right)^{\otimes 4} $ for a fixed stabilizer state $|\psi\rangle$. In order to solve Eq.~(\ref{cliffordintm1}), we first consider the following relation:
\begin{align}\label{cliffordint}
\int_{C}d\mu(C)&\left[ C^{\otimes 4}U^{\otimes 4}\overline{\left(|\psi\rangle\langle\psi |\right)^{\otimes 4}}U^{\dagger \otimes4}C^{\dagger \otimes4} \right]\nonumber\\
&=\sum_{\pi, \sigma} \left[ W^{+}_{\pi, \sigma}\text{Tr}\left( U^{\otimes 4}\overline{\left(|\psi\rangle\langle\psi |\right)^{\otimes 4}}U^{\dagger \otimes 4}QT_{\pi} \right) QT_{\sigma}+W^{-}_{\pi, \sigma}\text{Tr}\left( U^{\otimes 4}\overline{\left(|\psi\rangle\langle\psi |\right)^{\otimes 4}}U^{\dagger \otimes 4}Q^{\perp}T_{\pi} \right)Q^{\perp}T_{\sigma}  \right]\nonumber\\
&=\sum_{\pi, \sigma} \left[ W^{+}_{\pi, \sigma}\text{Tr}\left( U^{\otimes 4}\overline{\left(|\psi\rangle\langle\psi |\right)^{\otimes 4}}U^{\dagger \otimes 4}Q \right) QT_{\sigma}+W^{-}_{\pi, \sigma}\text{Tr}\left( U^{\otimes 4}\overline{\left(|\psi\rangle\langle\psi |\right)^{\otimes 4}}U^{\dagger \otimes 4}Q^{\perp}\right)Q^{\perp}T_{\sigma}  \right], 
\end{align} 
where $ Q^{\perp} = \mathbb{I}_{2^N} - Q $, and $ \{T_{\pi}\} $ denotes the set of permutation operators acting on four replicas of the Hilbert space that support both $ |\psi\rangle $ and $ U $. The coefficients $W^{\pm}_{\pi, \sigma}$ denote the generalized Weingarten functions. In the second equality, we made use of the facts $[T_{\pi}, U^{\otimes 4}]=0$ and $T_{\pi}|\psi\rangle^{\otimes 4}=|\psi\rangle^{\otimes 4}$ for all $\pi$ and $|\psi\rangle$. 
In Eq.~(\ref{cliffordint}), we notice the terms
\begin{eqnarray}
\text{Tr}\left[ U^{\otimes 4}\overline{\left( |\psi\rangle\langle\psi | \right)^{\otimes 4}} U^{\dagger \otimes4} Q \right] =\dfrac{1-m_{p}(U)}{2^N}
\end{eqnarray}
and 
\begin{eqnarray}\label{D5}
\text{Tr}\left[ U^{\otimes 4}\overline{\left( |\psi\rangle\langle\psi | \right)^{\otimes 4}} U^{\dagger \otimes4} Q^{\perp} \right]&=& 1-\text{Tr}\left[ U^{\otimes 4}\overline{\left( |\psi\rangle\langle\psi | \right)^{\otimes 4}} U^{\dagger \otimes4} Q \right]\nonumber\\
&=&m_{p}(U)+\left(1-m_{p}(U)\right)\left( 1-\dfrac{1}{2^N} \right).
\end{eqnarray}
In the second equality of Eq.~(\ref{D5}), we have added and subtracted $m_{p}(U)$ to obtain a form that is convenient for later. After incorporating the above equations in Eq.~(\ref{cliffordint}), we get
\begin{align}\label{cliffordint1}
\int_{C}d\mu(C)&\left[ C^{\otimes 4}U^{\otimes 4}\overline{\left(|\psi\rangle\langle\psi |\right)^{\otimes 4}}U^{\dagger \otimes 4}C^{\dagger \otimes4} \right]\nonumber\\
&=\left[ 1-m_{p}(U) \right] \sum_{\pi, \sigma}\left[ \dfrac{W^{+}_{\pi, \sigma}QT_{\sigma}}{2^N}+\left( 1-\dfrac{1}{2^N} \right)W^{-}_{\pi, \sigma} Q^{\perp}T_{\sigma} \right]+m_{p}(U)\sum_{\pi, \sigma}W^{-}_{\pi, \sigma}Q^{\perp}T_{\sigma}. 
\end{align}
The right-hand side of this equation can be simplified by noticing that
\begin{eqnarray}
\overline{\left( |\psi\rangle\langle\psi | \right)^{\otimes 4}}= \sum_{\pi, \sigma}\left[ \dfrac{W^{+}_{\pi, \sigma}QT_{\sigma}}{2^N}+\left( 1-\dfrac{1}{2^N} \right)W^{-}_{\pi, \sigma} Q^{\perp}T_{\sigma} \right]. 
\end{eqnarray}
Incorporating this relation in Eq.~(\ref{cliffordint1}), we get
\begin{eqnarray}
\int_{C}d\mu(C)\left[ C^{\otimes 4}U^{\otimes 4}\overline{\left(|\psi\rangle\langle\psi |\right)^{\otimes 4}}U^{\dagger 4}C^{\dagger \otimes4} \right]=\left( 1-m_{p}(U) \right)\overline{\left( |\psi\rangle\langle\psi | \right)^{\otimes 4}}+m_{p}(U)\sum_{\pi, \sigma}W^{-}_{\pi, \sigma}Q^{\perp}T_{\sigma}.  
\end{eqnarray}
Finally, substituting this equation in Eq.~(\ref{cliffordintm1}) gives 
\begin{eqnarray}\label{cliffordint2}
\langle m_{p}(VCU)\rangle_{C}&=&1-2^N(1-m_{p}(U))\text{Tr}\left[ V^{\dagger \otimes 4}QV^{\otimes 4}\overline{\left( |\psi\rangle\langle\psi | \right)^{\otimes 4}} \right]-2^Nm_{p}(U)\text{Tr}\left( V^{\dagger \otimes 4}QV^{\otimes 4}\sum_{\pi, \sigma}W^{-}_{\pi, \sigma}Q^{\perp}T_{\sigma} \right) \nonumber\\
&=&1-(1-m_{p}(U))(1-m_{p}(V))-m_{p}(U)2^N\text{Tr}\left( V^{\dagger \otimes 4}QV^{\otimes 4}\sum_{\pi, \sigma}W^{-}_{\pi, \sigma}Q^{\perp}T_{\sigma} \right)\nonumber\\
&=&m_{p}(U)+m_{p}(V)-m_{p}(U)m_{p}(V)-m_{p}(U)2^N\text{Tr}\left( V^{\dagger \otimes 4}QV^{\otimes 4}\sum_{\pi, \sigma}W^{-}_{\pi, \sigma}Q^{\perp}T_{\sigma} \right).
\end{eqnarray}
In the second equality, we used the fact $2^N\text{Tr}\left[ V^{\dagger \otimes 4}QV^{\otimes 4}\overline{\left( |\psi\rangle\langle\psi | \right)^{\otimes 4}} \right]=1-m_{p}(V)$. We now relate the last term on the right-hand side of Eq.~(\ref{cliffordint2}) to $m_p(V)$. To do so, we substitute a Haar-random unitary $\mathcal{U}$ from the unitary group $\mathcal{U}(2^N)$ in place of $U$ and compute $\langle\langle m_p(VCU)\rangle_{C}\rangle_{U\in \mathcal{U}(2^N)}$. It is straightforward to see that $\langle\langle m_p(VCU)\rangle_{C}\rangle_{U\in \mathcal{U}(2^N)}=\overline{m_p}$, where $\overline{m_{p}}$ denotes the Haar value for the non-stabilizing power and is given by $\overline{m_{p}}=1-4/(2^N+3)$ \cite{leone2022stabilizer}. It then follows that 
\begin{eqnarray}
 \overline{m_p}=\overline{m_p}+m_p(V)-\overline{m_p}m_p(V)-\overline{m_p}2^N\text{Tr}\left( V^{\dagger \otimes 4}QV^{\otimes 4}\sum_{\pi, \sigma}W^{-}_{\pi, \sigma}Q^{\perp}T_{\sigma} \right).
\end{eqnarray}
Consequently, we get
\begin{eqnarray}
2^N\text{Tr}\left( V^{\dagger \otimes 4}QV^{\otimes 4}\sum_{\pi, \sigma}W^{-}_{\pi, \sigma}Q^{\perp}T_{\sigma} \right)=m_{p}(V)\left( \dfrac{1}{\overline{m_{p}}}-1 \right).
\end{eqnarray}
Finally, Eq.~(\ref{cliffordint2}) becomes
\begin{eqnarray}\label{magicdec}
\langle m_{p}(VCU)\rangle_{C}=m_{p}(U)+m_{p}(V)-\dfrac{m_{p}(U)m_{p}(V)}{\overline{m_{p}}}.  
\end{eqnarray}
This concludes the proof of Theorem \ref{theo1}.  
\end{proof}\\

\section{Fluctuation of $m_p(VCU)$ and typicality of Theorem~\ref{theo1}}
\label{app-var}
{
In this appendix, we analyze the fluctuations of $m_p(VCU)$ (denoted with $\Delta m_p(VCU)$) when $C$ is drawn uniformly at random from the Clifford group $C(2^N)$. In particular, we derive an upper bound on the fluctuations and also provide exact values for special cases. For simplicity, we take $U$ and $V$ to be either identical. Then, the variance of $m_p(UCU)$ can be evaluated to be 
\begin{eqnarray}\label{F1}
\Delta^2m_p(VCU) = \langle m^2_p(UCU)\rangle_{C} - \langle m_p(UCU)\rangle^2_{C}. 
\end{eqnarray}
The second term on the right-hand side of the above equation is related to $m_p(U)$ via Eq.~(\ref{sameU}) of the main text as 
\begin{equation}\label{F2}
\langle m_p(UCU)\rangle^2_{C}=m^2_p(U)\left[ 2-\dfrac{m_p(U)}{\overline{m_p}} \right]^2.     
\end{equation}

On the other hand, we expand the first term by rewriting it in the form of Eq.~(\ref{25}), as follows: 
\begin{eqnarray}
\langle m^2_p(UCU)\rangle_{C} = 1&-&2^{N+1}\int_{C}d\mu(C)\text{Tr}\left[ U^{\dagger \otimes 4}QU^{\otimes 4} C^{\otimes 4}U^{\otimes 4}\overline{\left(|\psi\rangle\langle\psi |\right)^{\otimes 4}} U^{\dagger \otimes 4}C^{\dagger \otimes 4} \right] \nonumber\\
&+& 2^{2N} \int_{C}d\mu(C)\text{Tr}\left[ U^{\dagger \otimes 8}Q^{\otimes 2}U^{\otimes 8} C^{\otimes 8}U^{\otimes 8}\left\{\overline{\left(|\psi\rangle\langle\psi |\right)^{\otimes 4}}\right\}^{\otimes 2} U^{\dagger \otimes 8}C^{\dagger \otimes 8} \right], 
\end{eqnarray}
where the overline denotes the Haar average over the ensemble of stabilizer states. 
The above equation can be further simplified and can be written as
\begin{eqnarray}\label{F4}
\langle m^2_p(UCU)\rangle_{C} &=& 2\; \langle m_p(UCU)\rangle_{C}-1 +     2^{2N} \int_{C}d\mu(C)\text{Tr}\left[ U^{\dagger \otimes 8}Q^{\otimes 2}U^{\otimes 8} C^{\otimes 8}U^{\otimes 8}\overline{\left(|\psi\rangle\langle\psi |\right)^{\otimes 8}} U^{\dagger \otimes 8}C^{\dagger \otimes 8} \right]\nonumber\\
&=& 2\; m_p(U)\left[ 2-\dfrac{m_p(U)}{\overline{m_p}} \right] - 1 + \langle \mathcal{F}^2  \rangle_{C}\;, 
\end{eqnarray}
where $\mathcal{F} =   2^{N} \text{Tr}\left[ U^{\dagger \otimes 4}QU^{\otimes 4} C^{\otimes 4}U^{\otimes 4}\overline{\left(|\psi\rangle\langle\psi |\right)^{\otimes 4}} U^{\dagger \otimes 4}C^{\dagger \otimes 4} \right]$. Taking Eqs.~(\ref{F2}) and ~(\ref{F4}) together, the variance appearing in Eq.~(\ref{F1}) becomes
\begin{eqnarray}
\Delta^2 m_p(UCU) =  2\; m_p(U)\left[ 2-\dfrac{m_p(U)}{\overline{m_p}} \right] - 1 + \langle \mathcal{F}^2  \rangle_{C}\ - m^2_p(U)\left[ 2-\dfrac{m_p(U)}{\overline{m_p}} \right]^2.
\end{eqnarray} 
An upper bound on the variance follows from the fact that 
$\mathcal{F} = 1 - m_p(UCU) \leq 1$, which implies $\mathcal{F}^2 \leq \mathcal{F}$ 
and hence $\langle \mathcal{F}^2 \rangle_{C} \leq \langle \mathcal{F} \rangle_{C}$. From this, it immediately follows that 
\begin{eqnarray}\label{var-ineq}
\Delta^2 m_p(UCU)&\leq &  m_p(U)\left[ 2-\dfrac{m_p(U)}{\overline{m_p}} \right] - m^2_p(U)\left[ 2-\dfrac{m_p(U)}{\overline{m_p}} \right]^2\nonumber\\ 
&\leq & \left[ 2m_p(U)-\dfrac{m^2_p(U)}{\overline{m_p}} \right] \left[ 1-2m_p(U)+\dfrac{m^2_p(U)}{\overline{m_p}} \right] .     
\end{eqnarray}

When $U$ is a Clifford unitary, the right-hand side of the above inequality vanishes, implying that fluctuations are completely suppressed. Consequently, the bound is tight for unitaries with small non-stabilizing power and becomes exact in the Clifford limit. 

In contrast, when $U$ is a Haar random unitary, the variance can be analytically obtained by averaging over $U\in U(2^N)$. 
\begin{eqnarray}\label{var-haar}
 \Delta^2 m_p(UCU) = 2\overline{m_p}-\overline{m_p}^2-1+2^{2N}\text{Tr}\left( Q^{\otimes 2}\Pi^{(8)} \right).
\end{eqnarray}
In this case, when the system size ($N$) is large, the Haar averaged non-stabilizing power $\overline{m_p}$ approaches $1$, again leading to a suppression of fluctuations.

From the inequality in Eq.~(\ref{var-ineq}) and Eq.~(\ref{var-haar}), it can be seen that the fluctuations remain minimal when $U$ is either close to a Clifford or a Haar random unitary. 
This is indeed consistent with the behavior that we observe numerically in Figs.~\ref{fig:nonstab1} and ~\ref{fig:nonstab2} of the main text.

}

\section{Some consequences of Theorem \ref{theo1}}\label{app-f}
The above theorem has several direct consequences that are worth mentioning. 

\subsection{Corollary \ref{corrollaryc}: Non-stabilizing power when random Clifford operations are interlaced with non-Clifford operations }
\label{app-F1}

The result in Theorem~\ref{theo1} helps one to track the evolution of non-stabilizing power in circuits where the random Clifford operations are repeatedly interlaced with arbitrary non-Clifford operations. In particular, when all the non-Clifford operations have identical non-stabilizing powers, a closed-form expression for the evolution of $m_p$ follows:
\begin{eqnarray}\label{F1}
\langle M_{p}(U^{(t)})\rangle_{C_1, C_2, \cdots, C_{t-1}}=\overline{m_{p}}\left[ 1-\left( 1-\dfrac{m_{p}(U)}{\overline{m_{p}}} \right)^{t} \right]=\overline{m_p}\left[ 1-\exp\left\{ t\ln\left( 1-\dfrac{m_p(U)}{\overline{m_p}} \right) \right\} \right].
\end{eqnarray}
This result is also presented in Corollary \ref{corrollaryc} of the main text. 

\subsection{Thermalization of non-stabilizing power in limit of $m_p(U)/\overline{m_p}\ll 1$}
\label{app-F2}

In the main text, we have demonstrated the evolution of non-stabilizing power using a solvable two-qubit setting. It is interesting to probe how the thermalization of $m_p$ scales with the gate parameters, such as the Euler angles or the interaction strengths. For this purpose, we take the limit of small $m_{p}(U)$, i.e., $m_{p}(U)\ll \overline{m_{p}}$, in which we have 
\begin{equation}
\ln\left( 1-\dfrac{m_p(U)}{\overline{m_p}} \right) \approx -\dfrac{m_p(U)}{\overline{m_p}}.    
\end{equation} 
Then, Eq.~(\ref{expo_magic}) (or Eq.~\ref{F1}) can be written as 
\begin{eqnarray}
\langle m_{p}(U^{(t)}) \rangle_{\tilde{C}}\approx \overline{m_{p}}\left[ 1-\exp\left\{ -t\dfrac{m_{p}(U)}{\overline{m_{p}}} \right\} \right].
\end{eqnarray}
The smaller the non-stabilizing power of $U$, the longer it takes for $\langle m_{p}(U^{(n)}) \rangle_{\tilde{C}}$ to thermalize and reach the saturation value. The number of random Cliffords or time steps needed to drive $\langle m_{p}(U^{(n)}) \rangle_{\tilde{C}}$ towards the Haar average is therefore given by the scale $t^{*}\sim \overline{m_{p}}/m_{p}(U)$. In the simplest case when $U$ is a two-qubit gate of the form $U=\exp\{ -i\frac{c_z}{2}\sigma_z\otimes\sigma_z\}$, we have
\begin{eqnarray}
m_{p}(U)=\dfrac{\sin^2(2c_z)}{5} \approx \dfrac{4c^2_z}{5} \,\thickspace\text{for }\,c_z\ll 1\,.
\end{eqnarray} 
It then follows that 
\begin{eqnarray}
t^{*}\sim \left(\dfrac{3}{7}\right)\dfrac{5}{4c^2_z}\,, 
\end{eqnarray}
implying that the thermalization time scales with the interaction strength as $t\sim c^{-2}_z$. 
It is worth mentioning that when $m_p(U) < \overline{m_p}$, the evolution of the non-stabilizing power remains strictly monotonically increasing. This statement holds even when different non-Clifford unitaries are used in the circuit, provided their respective non-stabilizing powers are smaller than the Haar-averaged value. In contrast, if at least one non-Clifford unitary $U_i$ satisfies $m_p(U_i) > \overline{m_p}$, then the evolution does not remain monotonic.

\subsection{Time steps needed to reproduce $m_p(T)$}
\label{app-F3}

The $ T$-gate is a non-Clifford operation that is a key for universal fault-tolerant quantum computation and also for provable quantum advantage \cite{howard2017application, liu2024classical, PRXQuantum.6.010337, PhysRevLett.116.250501}. Assume one has access to random Clifford operations as well as a unitary gate $U$ with non-zero but potentially small non-stabilizing power. It is then interesting to identify the number of applications of the initial non-Clifford operator needed to reproduce the magic content equivalent to that of the $T$-gate. 

To estimate this number, we consider the average over a sequence of applications of $U$ and random Clifford operations, denoted as in main-text Sec.~\ref{sec-impact-a} by $U^{(t)}=U_{t}C_{t-1}U_{t-1}\cdots C_1U_1$, where here $U_i\equiv U$, $\forall i=1,\dots t$. 
Let us denote with $m_{p}(U)$ and $m_{p}(T)$ the non-stabilizing powers of $U$ and $T$, respectively.
From Eq.~(\ref{expo_magic}), we then obtain 
\begin{eqnarray}
m_{p}(T)= \overline{m_{p}}\left[ 1-\left( 1-\dfrac{m_{p}(U)}{\overline{m_{p}}} \right)^{t} \right].
\end{eqnarray}
This implies that 
\begin{eqnarray}
t=\dfrac{\ln\left( 1-\dfrac{m_p(T)}{\overline{m_{p}}} \right)}{\ln\left( 1-\dfrac{m_p(U)}{\overline{m_{p}}} \right)}\,. 
\end{eqnarray}
In the limit where $m_p(U)\ll{\overline{m_{p}}}$, the number of applications of $U$ needed to reproduce the non-stabilizing power of a single $ T$-gate thus scales linearly with $\frac{\overline{m_{p}}}{m_p(U)}$. 

{
\section{Numerical details}

In this appendix, we briefly outline the details concerning the numerical results presented in the main text, in particular those discussed in Secs.~\ref{sec-impact} and~\ref{sec-quant-chaos}.

\subsection{Non-stabilizing power}
We first describe the procedure used to evaluate the non-stabilizing power of arbitrary quantum evolutions. 
For two-qubit gates, the non-stabilizing power can be evaluated straightforwardly, as discussed in Appendix~\ref{app-b}, by averaging the non-stabilizerness generated when the gate acts on all 60 two-qubit stabilizer states. However, for larger system sizes, evaluating $m_p(U)$ becomes computationally expensive for two main reasons. First, a direct evaluation of $m_p(U)$ involves computing $\mathrm{Tr}(Q U^{\otimes 4} Q \bm{\Pi}^{(4)} U^{\dagger \otimes 4})$, which is generally difficult even for moderately large systems. Second, the number of stabilizer states grows exponentially as $\sim 2^N \prod_{j=1}^{N} (2^j+1)$, rendering full enumeration, similar to the case of two-qubit gates, infeasible.

Hence, instead of computing $m_p(U)$ exactly, we focus on obtaining an approximate estimate. Specifically, we evaluate $m_p(U)$ by sampling over random stabilizer states and averaging the resulting non-stabilizerness after the action of $U$:
\begin{eqnarray}
m_p(U) &=& \left\langle \mathcal{M}(U|\psi\rangle) \right\rangle_{|\psi\rangle \in \mathrm{STAB}(2^N)}\nonumber\\
&=&\left\langle 1-2^N\sum_{j=0}^{2^{2N}-1}\dfrac{1}{2^{2N}}\langle \psi|P_j|\psi\rangle^4 \right\rangle_{|\psi\rangle\in\mathrm{STAB}(2^N)}.
\end{eqnarray}
Here, $\mathcal{M}(|\phi\rangle)$ quantifies the non-stabilizerness of the state $|\phi\rangle$. Random stabilizer states are generated by applying random Clifford operators to the reference state $|0\rangle^{\otimes N}$, using the \href{https://pypi.org/project/qiskit/}{Qiskit} package~\cite{gadi_aleksandrowicz_2019_2562111}.

In our calculations, $\mathcal{M}(|\phi\rangle)$ is evaluated exactly by computing expectation values of all $2^{2N}$ Pauli strings. This approach enables efficient estimation of $m_p(U)$ for small and moderately sized systems. As the operator-space non-stabilizing power (OSNP) of a unitary $U$ is determined by $m_p(U)$, we evaluate the OSNP using the same sampling strategy. For numerical reproducibility, all simulations use a fixed master seed of 12345.

\subsection{Quantum chaos in Floquet circuits}

In this work, quantum chaos is diagnosed via the spectral statistics of the Floquet operator $\mathcal{U}$. Specifically, we compute the average adjacent level-spacing ratio
\begin{equation}
\langle r \rangle = \mathrm{Avg}\left\{ r_i \right\}_{i=1}^{2^N-2}, \qquad 
r_i = \frac{\min(d_i,d_{i+1})}{\max(d_i,d_{i+1})},
\end{equation}
where $d_i = \phi_{i+1}-\phi_i$ denotes the spacing between adjacent eigenphases $\{\phi_i\}$ of $\mathcal{U}$. This quantity provides an unfolding-free probe of spectral correlations. Chaotic dynamics are identified by $\langle r \rangle$ approaching the circular unitary ensemble (CUE) value $\langle r \rangle_{\mathrm{CUE}} = 0.596543$, while regular or integrable dynamics yield the Poisson value $\langle r \rangle_{\mathrm{Poisson}} = 0.386294$. The eigenphases are obtained by exact diagonalization of the single-period Floquet operator. The resulting spacing ratios are then averaged over many independent realizations of the Floquet circuit. 

To further analyze finite-size effects along the Id-SWAP edge, we perform a data-collapse analysis using an appropriate finite-size scaling form. Specifically, the average level-spacing ratio $\langle r \rangle$ is plotted as a function of the scaling variable $(J-J_c)N^{1/\nu}$, where $J$ is the tuning parameter, $J_c$ the critical point, and $\nu$ the correlation-length exponent. The scaling function is defined as
\begin{eqnarray}
f\big(N^{1/\nu} (J-J_c)\big) = N^{-\zeta/\nu}\langle r\rangle,
\end{eqnarray}
where $\langle r \rangle$ is evaluated for system size $N$. In our numerical simulations shown in Figs.~\ref{fig:magivsR}(a) and ~\ref{fig:magivsR1}, data from different system sizes collapse onto a single curve, demonstrating the consistency of this scaling form. 

}

\vspace{1cm}
\twocolumngrid 

\bibliographystyle{myunsrtnat}
\bibliography{references}

\end{document}